\newif\ifsubmit     %
\newif\ifllncs      %
\newif\ifexabs      %
\newif\ifblind      %

\submittrue

\ifllncs
  \documentclass[runningheads,a4paper]{llncs}

\else
  \documentclass[letterpaper,11pt,pdfa]{article}
  \usepackage[in]{fullpage}
\fi

\usepackage{iftex}
\ifPDFTeX
  \usepackage[utf8]{inputenc}
  \usepackage[noTeX]{mmap}
  \usepackage[T1]{fontenc}
\fi
\ifLuaTeX
  \usepackage{luatex85}
  \usepackage[noTeX]{mmap}
\fi

\usepackage{mdframed}
\usepackage{amsmath}
\usepackage{amsfonts}
\usepackage{amssymb}
\usepackage{amsthm}
\usepackage{color}
\usepackage{hhline}

\usepackage{appendix}
\usepackage{algorithm}
\usepackage{algpseudocode}
\usepackage{braket}
\usepackage{comment}
\usepackage{url}
\usepackage{bbm}
\usepackage{multicol}
\usepackage{mdframed}
\usepackage{multirow}

\usepackage[bookmarks]{hyperref}
\usepackage[nameinlink]{cleveref}
\hypersetup{
  colorlinks,%
  allcolors=blue
}

\usepackage{tikz}
\def\checkmark{\tikz\fill[scale=0.4](0,.35) -- (.25,0) -- (1,.7) -- (.25,.15) -- cycle;}

\ifllncs

  \spnewtheorem{claim}{Claim}{\bfseries}{\rmfamily}
  \crefname{claim}{claim}{claims}
  \Crefname{claim}{Claim}{Claims}
\else
  \newtheorem{theorem}{Theorem}[section]
  \newtheorem{definition}[theorem]{Definition}
  \newtheorem{remark}[theorem]{Remark}
  \newtheorem{lemma}[theorem]{Lemma}

  \newtheorem{claim}[theorem]{Claim}
  \newtheorem*{remark*}{Remark}
  
  \newtheorem{fact}[theorem]{Fact}

  \newtheorem*{theorem*}{Theorem}
  \newtheorem*{lemma*}{Lemma}

\usepackage[style=alphabetic,minalphanames=3,maxalphanames=4,maxnames=99,backref=true]{biblatex}
  \renewcommand*{\labelalphaothers}{\textsuperscript{+}}
  \renewcommand*{\multicitedelim}{\addcomma\space}
  \DeclareFieldFormat{eprint:iacr}{Cryptology ePrint Archive: \href{https://ia.cr/#1}{\texttt{#1}}}
  \DeclareFieldFormat{eprint:iacrarchive}{Cryptology ePrint Archive: \href{https://eprint.iacr.org/archive/#1}{\texttt{#1}}}
  \AtEveryBibitem{%
    \clearlist{address}
    \clearfield{date}
    \clearfield{isbn}
    \clearfield{issn}
    \clearlist{location}
    \clearfield{month}
    \clearfield{series}
 
    \ifentrytype{book}{}{%
      \clearlist{publisher}
      \clearname{editor}
    }
  }
\fi

\usepackage{appendix}
\usepackage{algorithm}
\usepackage{algpseudocode}
\usepackage{braket}
\usepackage{comment}
\usepackage{url}
\usepackage{multicol}
\usepackage{tikz}
\usetikzlibrary{arrows,automata,positioning}
\usepackage{caption}
\usepackage[caption=false]{subfig}
\usepackage[font=small,labelfont=bf]{caption}
\usepackage{comment}
\usepackage{pifont}
\usetikzlibrary{patterns}

\usepackage{enumitem}
\setlist[description]{noitemsep}
\setlist[enumerate]{noitemsep}
\setlist[itemize]{noitemsep}

\usepackage{soul, xcolor, xparse}
\makeatletter
  \ExplSyntaxOn
    \cs_new:Npn \white_text:n #1
    {
      \fp_set:Nn \l_tmpa_fp {#1 * .01}
      \llap{\textcolor{white}{\the\SOUL@syllable}\hspace{\fp_to_decimal:N \l_tmpa_fp em}}
      \llap{\textcolor{white}{\the\SOUL@syllable}\hspace{-\fp_to_decimal:N \l_tmpa_fp em}}
    }
    \NewDocumentCommand{\whiten}{ m }
    {
      \int_step_function:nnnN {1}{1}{#1} \white_text:n
    }
  \ExplSyntaxOff
  
  \NewDocumentCommand{ \varul }{ D<>{5} O{0.2ex} O{0.1ex} +m } {%
    \begingroup
    \setul{#2}{#3}%
    \def\SOUL@uleverysyllable{%
      \setbox0=\hbox{\the\SOUL@syllable}%
      \ifdim\dp0>\z@
      \SOUL@ulunderline{\phantom{\the\SOUL@syllable}}%
      \whiten{#1}%
      \llap{%
        \the\SOUL@syllable
        \SOUL@setkern\SOUL@charkern
      }%
      \else
      \SOUL@ulunderline{%
        \the\SOUL@syllable
        \SOUL@setkern\SOUL@charkern
      }%
      \fi}%
    \ul{#4}%
    \endgroup
  }
\makeatother

\newcommand{\Z}{\mathbb{Z}}
\newcommand{\I}{\mathbb{I}}

\usepackage{mleftright}

\newcommand{\As}{\mathcal{A}}
\newcommand{\Bs}{\mathcal{B}}

\newcommand{\Ds}{\mathcal{D}}

\newcommand{\Es}{\mathcal{E}}
\newcommand{\Ts}{\mathcal{T}}

\newcommand{\Hs}{\mathcal{H}}

\renewcommand{\P}{\sf P}

\newcommand{\prg}{{\sf PRG}}

\newcommand{\enc}{{\sf Enc}}
\newcommand{\dec}{{\sf Dec}}

\newcommand{\gen}{{\sf Gen}}
\newcommand{\skgen}{{\sf SKGen}}
\newcommand{\pkgen}{{\sf PKGen}}

\newcommand{\A}{\mathsf{A}}
\newcommand{\B}{\mathsf{B}}
\newcommand{\C}{\mathsf{C}}

\newcommand{\Q}{\mathsf{Q}}
\newcommand{\R}{\mathsf{R}}

\newcommand{\F}{\mathbb{F}}
\newcommand{\E}{\mathsf{E}}
\newcommand{\D}{\mathsf{D}}
\newcommand{\M}{\mathsf{M}}
\newcommand{\N}{\mathsf{N}}

\newcommand{\ct}{{\sf ct}}

\mdfdefinestyle{figstyle}{ 
  linecolor=black!7, 
  backgroundcolor=black!7, 
  innertopmargin=10pt, 
  innerleftmargin=25pt, 
  innerrightmargin=25pt, 
  innerbottommargin=10pt 
}

\newcommand{\sk}{{\sf sk}}

\newcommand{\pk}{{\sf pk}}

\newcommand{\msg}{{\sf msg}}

\renewcommand{\kappa}{\ell}

\newcommand{\poly}{{\sf poly}}
\newcommand{\negl}{{\sf negl}}

\DeclareMathOperator{\Tr}{Tr}

\newcommand{\supp}{\mathsf{SUPP}}

\newcommand{\p}{\mathbf{p}}

\newcommand{\view}{{\sf View}}

\newcommand{\Wsf}{\mathsf{W}}
\newcommand{\Msf}{\mathsf{M}}

\title{
How (not) to Build Quantum PKE in Minicrypt
}
\author{
Longcheng Li\thanks{State Key Lab of Processors, Institute of Computing Technology, Chinese Academy of Sciences. Email: \texttt{lilongcheng22s@ict.ac.cn}} \quad 
Qian Li\thanks{Shenzhen International  Center For Industrial  And  Applied  Mathematics, Shenzhen Research Institute of Big Data. Email: \texttt{liqian.ict@gmail.com}}
\quad
Xingjian Li\thanks{Tsinghua University. Email: \texttt{lxj22@mails.tsinghua.edu.cn}}
\quad
Qipeng Liu\thanks{University of California San Diego. Email: \texttt{qipengliu0@gmail.com}}
}

\date{}

\begin{document}

\maketitle

\begin{abstract}
The seminal work by Impagliazzo and Rudich (STOC'89) demonstrated the impossibility of constructing classical public key encryption (PKE) from one-way functions (OWF) in a black-box manner. Quantum information has the potential to bypass classical limitations, enabling the realization of seemingly impossible tasks such as quantum money, copy protection for software, and commitment without one-way functions. However, the question remains: can quantum PKE (QPKE) be constructed from quantumly secure OWF?

A recent line of work has shown that it is indeed possible to build QPKE from OWF, but with one caveat. These constructions necessitate public keys being quantum and unclonable, diminishing the practicality of such ``public’’ encryption schemes --- public keys cannot be authenticated and reused. In this work, we re-examine the possibility of perfect complete QPKE in the quantum random oracle model (QROM), where OWF exists.

\begin{itemize}[topsep=1pt]
\item[] 
Our first main result: QPKE with classical public keys, secret keys and ciphertext, does not exist in the QROM, if the key generation only makes classical queries. 
\end{itemize}

Therefore, a necessary condition for constructing such QPKE from OWF is to have the key generation classically ``un-simulatable’’. Previous results (Austrin~et al. CRYPTO'22) on the impossibility of QPKE from OWF rely on a seemingly strong conjecture. Our work makes a significant step towards a complete and unconditional quantization of Impagliazzo and Rudich’s results. 

Our second main result extends to QPKE with quantum public keys.

\begin{itemize}[topsep=1pt]
\item[] 
The second main result: QPKE with quantum public keys, classical secret keys and ciphertext, does not exist in the QROM, if the key generation only makes classical queries and the quantum public key is either pure or ``efficiently clonable''.
\end{itemize}

The result is tight due to these existing QPKEs with quantum public keys, classical secret keys, quantum/classical ciphertext and classical-query key generation require the public key to be mixed instead of pure; or require quantum-query key generation, if the public key is pure. Our result further gives evidence on why existing QPKEs lose reusability. 

We also explore other sufficient/necessary conditions to build QPKE from OWF. Along the way, we use a new argument based on conditional mutual information and Markov chain to reprove the classical result; leveraging the analog of quantum conditional mutual information and quantum Markov chain by Fawzi and Renner (Communications in Mathematical Physics), we extend it to the quantum case and prove all our results. We believe the techniques used in the work will find many other usefulness in separations in quantum cryptography/complexity. 
\end{abstract}

\section{Introduction}

Quantum information and computation has the remarkable capability to transform classical impossibility into reality, ranging from breaking classically secure cryptosystems (Shor'a algorithm~\cite{shor1999polynomial}), realizing classically impossible primitives (quantum money~\cite{wiesner1983conjugate}, quantum copy-protection~\cite{aaronson2009quantum,aaronson2021new}) to weakening assumptions (quantum key distribution~\cite{Bennett_2014}, oblivious transfer/multi-party computation~\cite{bartusek2021round,grilo2021oblivious}, commitment~\cite{ananth2022cryptography,morimae2022quantum})\footnote{Here, we only cite works that initialized each area.}.

In the seminal work by Impagliazzo and Rudich~\cite{impagliazzo1989limits}, they proved that one-way functions (OWFs) were insufficient to imply the existence of public-key encryption (PKE) in a black-box manner. Coined by Impagliazzo~\cite{impagliazzo1995personal}, the word ``Minicrypt'' was referred to a world where only one-way functions exist; this word now broadly denotes all cryptographic primitives that are constructible from one-way functions. Thus, their result is now often interpreted as ``classical PKE is not in Minicrypt''. Given the increasing instances of quantum making classical impossibility feasible, we explore the following question in this work: 

\begin{center}
{\it 
Does quantum PKE (with classical plaintext) exist in Minicrypt?
}
\end{center}

Upon posing the question, ambiguity arises. A general quantum public-key encryption (QPKE) scheme allows everything to be quantum: its interaction with an OWF, both its public key and secret key, as well as ciphertext. Indeed, many efforts have already been made towards understanding different cases. 

\paragraph{Classical Keys, Classical Ciphertext.} 
In the work by Austrin, Chung, Chung, Fu, Lin and Mahmoody~\cite{austrin2022impossibility}, they initialized the study on the impossibility of quantum key agreement (QKA) in the quantum random oracle model. QKA is a protocol that Alice and Bob can exchange classical messages in many rounds, quantumly query a random oracle, and eventually agree on a classical key. They show that, under a seemingly strong\footnote{{It seems strong because it implies that the eavesdropper can attack quantum Alice and Bob with only a polynomial number of classical queries.}} assumption called ``polynomial compatibility conjecture'', such QKA with perfect completeness does not exist. Since QPKE with both keys and ciphertext being classical implies a two-round QKA, their conditional impossibility extends to this type of QPKE as well. 

The result provides evidence on the negative side: such QPKE does not exist under the polynomial compatibility conjecture. However, not only proving or refuting the conjecture is quite challenging, but also the conjecture (or some form of the conjecture) is necessary. To prove QKA is impossible, one needs to design an eavesdropper that observes the whole transcript, interacts with the random oracle, and guesses a classical key. In a general QKA, both Alice and Bob can make quantum queries to the random oracle; the eavesdropper in the attack by Austrin et al. only makes classical queries. It inherently requires a simulation of a quantum-query algorithm using only classical queries (at least for the QKA functionality). Although not directly comparable, such efficient simulation for decision problems~\cite{aaronson2009need} (Aaronson-Ambainis conjecture) is conjectured to exist but the question still remains open until now.  Moreover, \cite{austrin2022impossibility} showed that if the Aaronson-Ambainis conjecture is false, a classical-query eavesdropper is insufficient to break imperfect complete QKA protocols.

In light of these considerations, our focus in this work centers on the following question:
\begin{center}
    {\it Q1. Can we separate QPKE with classical keys, classical ciphertext from Minicrypt, \\ without any conjecture?}
\end{center}

\paragraph{Quantum Public Key, Classical Secret Key and Quantum/Classical Ciphertext.} 
The quantum landscape introduces a paradigm shift. It was first realized by Morimae and Yamakawa~\cite{morimae2022one} that some forms of QPKE with quantum public keys and quantum ciphertext might be constructed from OWFs (or even presumably weaker primitives). Subsequent QPKE schemes were later proposed by Coladangelo~\cite{coladangelo2023quantum}, Kitagawa et al.~\cite{kitagawa2023quantum}, Malavolta and Walter~\cite{cryptoeprint:2023/500} and Barooti et al.~\cite{barooti2023publickey}. 

On the surface, it seems to give a good answer: QPKE with quantum public keys exists in Minicrypt, showcasing a notable distinction between the quantum and classical world. However, is this demarcation as unequivocal as it seems?  Indeed, all the aforementioned constructions share one limitation --- they lose one of the most important properties inherent in all classical PKE --- reusability. A classical PKE scheme allows a user who possesses a public key to encrypt any polynomial number of messages, by reusing the classical public key. In contrast, public keys in \cite{morimae2022one,coladangelo2023quantum,kitagawa2023quantum,cryptoeprint:2023/500,barooti2023publickey} are essentially unclonable\footnote{Although unclonability does not necessarily imply no-reusability, it is an evidence of no-reusability.} and not reusable. Although some may argue for the practicality of QPKE with an additional public interface for generating quantum public keys, it introduces additional complexities such as the authentication of quantum keys, leading to increased interactions and other potential challenges.

Thus, our second focus is the question:
\begin{center}
    {\it Q2. Must reusability be sacrificed in constructing QPKE with quantum public keys from OWFs?}
\end{center}

\subsection{Our Main Results}
We make progress towards these two questions. A table discussing and comparing all existing results and our results is provided on the next page~(\Cref{tab:all_results}). 
Our first main result establishes an impossibility result on QPKE with classical keys and classical ciphertext in the QROM. 
\begin{theorem}\label{thm:main_theorem1}
    QPKE with classical keys and classical ciphertext does not exist in the QROM, if 
    \begin{enumerate}
        \item It has perfect completeness.
        \item The key generation algorithm only makes classical queries to the random oracle. 
    \end{enumerate}
\end{theorem}
Unlike the approaches used in \cite{austrin2022impossibility}, our result does not require any conjecture, as the eavesdropper in our attack makes quantum queries. Perfect completeness is a natural property shared by many PKE schemes\footnote{Very recently, Mazor~\cite{mazor2023key} proposed the first perfectly complete ``Merkle Puzzle'', which was not known for many years.}

Consequently, if QPKE could be constructed in the QROM, one must have the key generation procedure being classical-query ``un-simulatable''; meaning the keys can not be computed by an efficient classical-query algorithm. If we believe the Aaronson-Ambainis conjecture~\cite{aaronson2009need} (all quantum-query-solvable decision problems can be simulated by classical-query), the theorem suggests that the key generation  must be a sampling procedure. Thus, to build QPKE in the QROM, we need to find a search/sampling problem that is only tractable by quantum queries; one such example is the Yamakawa-Zhandry problem~\cite{yamakawa2022verifiable}. 

\cite{austrin2022impossibility} ruled out perfect complete QPKE with classical keys and classical ciphertext in the QROM, if (i) the encryption only makes classical queries, or (ii) both the key generation and decryption makes classical queries, without using the polynomial compatibility conjecture. Our result immediately improves their result (ii). Additionally, both of their impossibility results apply to the case where the oracle access in encryption and decryption procedure is asymmetric. Our result further completes the picture by making both encryption and decryption symmetric (quantum queries). 

Another difference in our work is that our eavesdropper makes quantum queries. The capability of leveraging quantum eavesdroppers and conditioning on quantum events in this work, immediately gives us two strengthened versions (\Cref{thm:main_theorem2} and \Cref{thm:main_theorem3}) of \Cref{thm:main_theorem1}, which works for quantum public keys and quantum ciphertext. These two extensions provide many interesting discussions on the feasibility and impossibility of building QPKE in Minicrypt. \cite{grilo2023towards} also discussed the feasibility of QPKE with quantum ciphertext, with a weaker result and based on the polynomial compatibility conjecture. Thus, we believe our framework is versatile and has the potential to completely answer these questions. We elaborate on them now. 

\vspace{1em}

Our second main result extends the previous theorem to QPKE with quantum public keys and ciphertexts. This scheme is tight to all existing QPKE with quantum public keys and perfect completeness, see \Cref{tab:comparison}.
\begin{theorem}[Subsuming \Cref{thm:main_theorem1}]\label{thm:main_theorem2}
    QPKE with quantum public key, classical secret key, and classical or quantum ciphertext does not exist in the QROM, if 
    \begin{enumerate}
        \item It has perfect completeness.
        \item The key generation algorithm only makes classical queries to the random oracle. 
        \item The quantum public key is either pure or ``efficiently clonable''. 
    \end{enumerate}
\end{theorem}
At first glance, the theorem may appear unexciting: how can a classical-query key generation produce meaningful quantum public keys? Interestingly, the QPKE constructions in \cite{kitagawa2023quantum,cryptoeprint:2023/500} have only classical-query key generation procedures. Namely, their public keys are of the form $\frac{1}{\sqrt{2}} (\ket {s_0} + \ket {s_1})$ of two strings $s_0, s_1$, which can be computed using only classical queries. 

Let us explain the third condition in our result. In QPKE, one needs to guarantee that having multiple copies of a quantum public key, the IND-CPA security still holds. If the public key generation procedure outputs a pure state, in other words the security holds against copies of the same pure state, we call the public key pure. Otherwise, each copy of a quantum public key is a mixed state (a distribution of pure states) and even holding multiple copies from the same distribution does not mean having the same pure states\footnote{For example, even if having two copies of the same mixed state $\frac{\mathbb{I}}{2}$, they can be two different states, like $\ket 0 \ket 1$.}. In the latter case, we require that there exists a query-efficient cloning procedure that can perfectly duplicate the pure state of the quantum public key.

As the schemes in \cite{kitagawa2023quantum,cryptoeprint:2023/500} have perfect completeness and classical-query key generation, our result essentially says that their public keys must be mixed and can not be cloned. While the absence of perfect cloning does not imply non-reusability, our result provides crucial insights into the feasibility of building QPKE with quantum public keys and reusability from OWFs. Many intriguing open questions will be discussed further at the end of this section. 

Our second theorem also is tight to the QPKE scheme by \cite{coladangelo2023quantum,barooti2023publickey}. Their scheme has perfect completeness, pure state quantum public keys. Our theorem suggests that their schemes must make quantum queries; which is the case in their construction, as they need to query a pseudorandom function (implied by OWFs in a black-box manner) on an equal superposition.

\vspace{1em}

Our next result is on the type of QPKE with classical keys and quantum ciphertext, whose decryption makes no queries to an oracle. The impossibility extends to any oracle model (even quantum oracles), not just the random oracle model.  %

\begin{theorem}\label{thm:main_theorem3}
    (Imperfect/Perfect) QPKE with classical keys and quantum ciphertext does not exist \emph{in any oracle model}, if 
    \begin{enumerate}
        \item The decryption algorithm makes no queries to the oracle. 
    \end{enumerate}
\end{theorem}
Here, an oracle model means an oracle is sampled from a distribution, and the KA is executed under this oracle; both the key generation algorithm, the encryption algorithm and the attacker can have quantum access to the oracle. We remark that in our impossibility result, such QPKE does not need to be perfectly complete. We also note that, our theorem holds in the classical setting; PKE does not exist in any oracle model if the decryption algorithm makes no queries. This does not contradict constructions in the generic group model, as the decryption algorithm is required to make oracle queries. 

A recent work by Bouaziz–Ermann, Grilo, Vergnaud and Vu ~\cite{grilo2023towards} also discusses a similar separation. They proved that, under the polynomial compatibility conjecture, perfect complete QPKE with classical keys and quantum ciphertext does not exist in the QROM, if decryption makes no queries. Our \Cref{thm:main_theorem3} improves their results in three aspects: we remove the conjecture and the requirement on perfect completeness, and the impossibility works in any oracle model.

\paragraph{Other results.} As we progress towards achieving our main results, our techniques also enable us to establish additional impossibility 
results on QPKE and QKA as well. These results include
\begin{itemize}
\item Separation between pseudorandom quantum states and QKA (Remark \ref{rem:prs});
\item Impossibility of Merkle-like QKA (more generally, non-interactive QKA) in \emph{any oracle model} (Theorem \ref{thm:non-interactive}); we emphasize that this result also holds for quantum oracles, e.g., when a black-box unitary is chosen from Haar measure. 

\item Impossibility of QPKE with a short classical secret key in \emph{any oracle model} (Appendix \ref{sec:shortkey}).
\end{itemize}
We refer interested readers to these sections for more details.

\begin{table}[!htb]
  \centering
  \begin{tabular}{|l||c|c|c|c||c|c|}
    \hline
     & \cite{austrin2022impossibility} & \cite{austrin2022impossibility}  & \cite{austrin2022impossibility} & Thm \ref{thm:main_theorem1} & \cite{grilo2023towards} & Thm \ref{thm:main_theorem3} \\
    \hline
    perfect complete & \checkmark & \checkmark & \checkmark & \checkmark & \checkmark &   \\
    \hline
    ciphertext & C & C & C & C & Q & Q\\
    \hline
    ${\sf Gen}$ & Q & Q & C & C & Q & Q\\
    \hline
    ${\sf Enc}$ & Q & C & Q & Q & Q & Q \\
    \hline
    ${\sf Dec}$ & Q & Q & C & Q & $\star$ & $\star$ \\
    \hline
    conjecture & \checkmark &  &  & & \checkmark & \\
    \hline
    oracle & RO & RO & RO & RO & RO & any oracle \\
    \hline
  \end{tabular}
  \caption{Comparing impossibility results on classical public keys. `Q' denotes quantum, `C' denotes classical, `checkmark' denotes yes, and `$\star$' denotes no oracle queries. ``RO'' stands for ``(quantum) random oracle''. }
  \label{tab:all_results}
\end{table}

\begin{table}[!htb]
  \centering
  \begin{tabular}{|l||c|c|c|}
    \hline
      & Thm \ref{thm:main_theorem2} (Impossibility) & \cite{kitagawa2023quantum,cryptoeprint:2023/500}  & \cite{coladangelo2023quantum,barooti2023publickey} \\
    \hline
    perfect complete & \checkmark & \checkmark & \checkmark  \\
    \hline
    ciphertext & pure or ``efficiently clonable'' & mixed and unclonable & pure\\
    \hline
    ${\sf Gen}$ & C & C & Q\\
    \hline
  \end{tabular}
  \caption{Comparing our Thm \ref{thm:main_theorem2} with existing constructions. \cite{kitagawa2023quantum,cryptoeprint:2023/500} overcame the impossibility by allowing the ciphertext to be mixed; \cite{coladangelo2023quantum,barooti2023publickey} overcame the impossibility by allowing the Gen procedure to make quantum queries. }
  \label{tab:comparison}
\end{table}

\subsection{Open Questions and Discussions}

Before the overview of our techniques in the next section, we discuss some  open questions and directions. Some discussions may become clearer upon reviewing both the overview and the entirety of the paper.

\paragraph{Remove perfect completeness.} Our \Cref{thm:main_theorem1} and \Cref{thm:main_theorem2} both rely on the underlying QPKE is perfectly complete. Can we remove this condition and make the impossibility result work for any QPKE? %

\paragraph{Extending clonable keys to reusable keys.}
When quantum public keys are reusable, that means there exists an encryption procedure $\enc'$ that takes two messages $m, m'$ and only one quantum public key $\pk$, and outputs two valid ciphertexts. Since $\enc'$ is  not necessarily honest, it seems difficult to convert such assumptions into the setting of a QKA protocol. An intermediate goal is to have $\enc$ ``deterministic'', which means it takes $\pk$ and $m$, outputs the original $\pk$ together with a valid $\ct$ --- this is a weaker but promising step towards a full understanding of reusability.
Another direction is to relax the ``clonable'' condition to some form of ``approximate clonable'' in our impossibility result.

\paragraph{Make key generation quantum.} Our \Cref{thm:main_theorem1} and \Cref{thm:main_theorem2} assume key generation procedure only makes classical queries. The current techniques require addressing some sorts of ``heavy queries'' made by an encryption procedure. For classical-query key generation, the number of such possible inputs is only polynomial. But it will be less clear how to handle this, when the key generation is fully quantum. Upon resolving this, one can get a full separation between QPKE and OWFs. 

The basic idea of our approach is to keep the conditional mutual information (between Alice and Bob, conditioned on quantum Eve) small during the execution of QKA protocols or QPKE schemes. The characterization of quantum states with small CMI provided by Fawzi and Renner \cite{fawzi2015quantum} (see  \Cref{thm:quantum_mutual_info_operational}) plays a major role. To achieve a full separation between QPKE and OWF, the argument has to critically utilize the condition that $H$ is a uniform random function rather than drawn from a general oracle distribution. This is because our separation works in both the quantum and classical setting, but in the full classical KA case, KA exists in the generic group model. 
Thus, we really need to use the structures of a random oracle. In this work, we show it is possible when key generation only makes classical queries but not clear for quantum-query key generation.

\subsection*{Acknowledgement}
Qipeng Liu would like to thank Navid Alamati for bringing up this problem in 2019 while interning at Fujitsu Research of America, and Jan Czajkowski, Takashi Yamakawa and Mark Zhandry for insightful discussions at Princeton during 2019 and 2020. The authors also thank Luowen Qian for pointing out the problem of Merkle-like QKA with Haar-random unitary, which prompted the realization of Theorem \ref{thm:main_theorem3} applicable to even quantum oracles.

Qian Li's work was supported by the National Natural Science Foundation of China Grants (No. 62002229), and Hetao
Shenzhen-Hong Kong Science and Technology Innovation Cooperation Zone Project (No. HZQSWS-KCCYB-2024016).

Longcheng Li was supported in part by the National Natural Science Foundation of China Grants No. 62325210, 62272441, 12204489, 62301531, and the Strategic Priority Research Program of Chinese Academy of Sciences Grant No. XDB28000000.

\section{Technical Overview}

In this section, we shed light on several key ideas used in the work, especially those on how to attack with quantum queries and how to deal with quantum public keys. Our method for handling quantum ciphertext is similar to that in~\cite{grilo2023towards} by leveraging the Gentle Measurement Lemma. 

\paragraph{Translating QPKE to QKA.}
Although we focus on the separation between QPKE and OWFs (or the quantum random oracle model), we will mostly study two-round QKA protocols. Say, if we have a QPKE scheme in the QROM ($\gen, \enc, \dec$), we can easily convert it to a two-round QKA: 
\begin{definition}[Informal, Two-round  QKA]\label{def:informal_QKA}
We define a two-round QKA as follows.
\begin{itemize}
    \item Alice runs $\gen$ to produce $\pk, \sk$ and sends $\pk$ as the first message $m_1$ to Bob. We call Alice's algorithm at this stage $\As_1$. 
    \item Bob upon receiving $m_1:=\pk$, it samples a uniformly random key $k$. It computes $\ct \gets \enc(\pk, k)$ and  sends $m_2 := \ct$ to Alice. We call Bob's algorithm at this stage $\Bs$. 
    \item Alice will then run $\dec(\sk, \ct)$ to retrieve $k$. We call Alice's algorithm at this stage $\As_2$. 
\end{itemize}
\end{definition}
If the QPKE has security against randomly chosen messages, then the underlying QKA is secure. Quantum public keys will make $m_1$ quantum in the corresponding QKA, quantum ciphertext translates to a quantum $m_2$ and classical-query $\gen$ makes $\As_1$  classical-query. 
Thus, we focus on breaking various types of two-round QKA with perfect completeness.  

\subsection{Recasting the Classical Idea for Merkle-like KA}

The classical proofs~\cite{impagliazzo1989limits,barak2009merkle} shares one common idea: as long as an eavesdropper Eve learns all the queries that are both used in Alice and Bob's computation, Eve can learn the key\footnote{\cite{brakerski2011limits} also gave a proof for the case of perfect completeness. Their proof works even if not all intersection queries are learned.}. Here we re-interpret the idea for a special case: two-round Merkle-like KA. In this type of KA, $m_1, m_2$ are generated based on oracle queries and sent simultaneously to the other party (so that $m_2$ has no dependence on $m_1$). Furthermore, to recover the shared key, Alice and Bob only need to do local computation on their internal states and the communication, \emph{without making any oracle queries}. This is a strong form of KA, which is not implied by PKE using the aforementioned reduction. We refer to it as Merkle-like because the famous Merkle Puzzles are of this form.

Let us assume for a specific execution of Alice and Bob, the queries made by Alice is the list $R_A$, her private coins are $s_A$ and similarly $R_B, s_B$, $R_E$ is the query made by Eve; we define $\view_A = (s_A, R_A)$ as the personal view of Alice, similarly $\view_B$ for Bob.
A tuple $(\view_A, \view_B, m_1, m_2, H)$ is a possible execution of the KA right after Alice and Bob exchange their messages but have not started working on computing the key.
The tuple specifies Alice's random coins and queries, similarly for Bob, communication $m_1, m_2$ and the oracle $H$ under which the protocol is executed. 
We say $(\view_A, \view_B)$ is consistent with a transcript $(m_1, m_2)$ if there exists an oracle $H$, such that $(\view_A, \view_B, m_1, m_2, H)$ has \emph{strictly positive} probability of appearing in some real execution. Similarly, $\view_A$ is consistent with $(m_1, m_2)$ if $(\view_A, m_1, m_2, H)$ has non-zero probability for some oracle $H$.

For a transcript $(m_1, m_2)$, it is always easy to find a pair of Alice's fake $\view_A' = (s'_A, T'_A)$ that is consistent with $m_1, m_2$. Since we do not care about the actual computation cost, and only the number of queries matters, we can keep sampling oracles until we see a transcript $(m_1, m_2)$. Upon receiving $\view_A'$, we hope that $(\view_A', \view_B)$ has the same distribution as $(\view_A, \view_B)$; if that is true, based on perfect completeness, any non-zero support in $(\view_A, \view_B)$ should provide us with the agreed key.

Unfortunately, this does not hold true. The overlap in inputs between $R_A$ and $R_B$ leads to a correlation in the distribution $(\view_A, \view_B)$, making it impossible to sample independently. This correlation is the resource for Alice and Bob to compute an agreed key, akin to Merkle Puzzles.

The key insight from \cite{impagliazzo1989limits,barak2009merkle} is to eliminate this correlation, often referred to as ``intersection queries''. In their attacks, Eve queries the oracle in a manner such that $R_E$ encompasses all the shared knowledge between Alice and Bob. More precisely, for any execution $R_A, R_B$ consistent with $(m_1, m_2)$, it holds that $R_A \cap R_B \subseteq R_E$. Consequently, conditioned on $R_E$ and the transcript, the distribution $(\view_A, \view_B)$ becomes close to a product distribution. Eve can then sample $\view_A'$ conditioned on $m_1, m_2, R_E$, and this fake view will yield the correct key.

Finally, \cite{impagliazzo1989limits,barak2009merkle} demonstrate that as long as Eve queries Alice and Bob's ``heavy queries'' (those with a relatively noticeable probability of being queried), with high probability $R_E$ will contain all intersection queries.

\subsection{A Quantization Attempt by Austrin~et al.}

\cite{austrin2022impossibility} focuses on the fully general QKA, but here we explain their ideas for the Merkle-like protocols. The first challenge arises when attempting to formally describe an execution for both quantum Alice and Bob. As they can make quantum queries, neither the random coins $s_A$ nor the query list $R_A$ can be explicitly delineated. A quantum algorithm can possess randomness that is impossible to be purified as random coins, and can make quantum superposition queries. An execution right after Alice and Bob exchange messages, is represented by $(\view_A, \view_B, m_1, m_2, H)$ where $\view_A := \rho_A, \view_B := \rho_B$ are the internal quantum states of Alice and Bob.

To quantize the strategy of \cite{impagliazzo1989limits,barak2009merkle}, one needs to define ``heavy quantum queries'' and ``quantum intersection queries'', and establish some form of independence between Alice and Bob conditioned on transcripts and Eve's knowledge. Austrin~et al., leveraged the breakthrough technique (``compressed oracle'') by Zhandry~\cite{zhandry2019record}, defined ``heavy quantum queries''. On a very high level (without introducing Zhandry's technique), ``heavy quantum queries'' are classical inputs that have high weights on the oracle, when the oracle is examined under the Fourier basis. 
Their proposed attack queries all ``heavy quantum queries'' classically. Under their polynomial compatibility conjecture, they can argue the success of their attacks. 

This approach is less ideal in the following aspects. First, even for Merkle-like protocols, the polynomial compatibility conjecture seems necessary. Second, the classical-query Eve has to somehow ``simulate'' the ability of quantum Alice or quantum Bob. Although such simulation is believed to be true for decision problems~\cite{aaronson2009need}, this kind of simulation in the QKA setting is both unclear to hold and potentially unnecessary. Finally, establishing independence between two quantum states (quantum Alice and Bob) is challenging to define and can be intricate.

\subsection{Step Back --- Classical and Quantum Proofs Using Markov Chain}

Stepping back, let's reconsider if there are other approaches that are more quantum-friendly: potentially can take advantage of quantum-query Eve. The key insight in \cite{impagliazzo1989limits,barak2009merkle} is, when conditioned on some query list $R_E$ and $(m_1, m_2)$, the distribution $(\view_A, \view_B)$ is a product distribution, meaning Alice and Bob are independent. Based on this, Eve can therefore sample $\view_A'$ and they together with the real $\view_B$ have the same distribution as $(\view_A, \view_B)$. 

Our first contribution is to give an alternative view of the classical proofs for Merkle-like KA.
We realize that, in the classical proof, when intersection queries $R_A \cap R_B$ always is in $R_E$, the conditional mutual information (CMI) between $R_A, R_B$ conditioned on $E, m_1, m_2$ is 0. 
\begin{align*}
    I(\view_A: \view_B| (m_1, m_2, R_E)) = 0.
\end{align*}
There are two seemingly classically equivalent consequences when CMI is 0. 
\begin{itemize}
    \item \textbf{Perspective 1}. From Pinsker's inequality, a CMI of 0 immediately implies that $\view_A$ and $\view_B$ are independent conditioned on $(m_1, m_2, R_E)$. This further implies that we can sample $\view_A'$ accordingly.

    \item \textbf{Perspective 2}. Another perspective is that, $\view_B \rightarrow (m_1, m_2, R_E) \rightarrow \view_A$ forms a Markov chain. 

    Three random variables $XYZ$ form a Markov chain if $p_{xyz} = p_{xy} p_{z|y}$. In our case, it says there exists a way to take $(m_1, m_2, R_E)$ as inputs and sample $\view'_A$ such that $(\view'_A, \view_B)$ and $(\view_A, \view_B)$ are identically distributed. 
\end{itemize}
The above discussion explains the attacks for Merkle-like protocols, through the lens of CMI. 

\vspace{1em}

From this, we propose one candidate quantum attack for Merkle-like QKA: 
\begin{mdframed}
    \textbf{Quantum Eve: A Framework} 
    \begin{itemize}
    \item Eve makes some quantum queries and let $\view_E := \rho_E$ be its internal quantum state. It ``somehow'' makes sure that the CMI between Alice and Bob, conditioned on Eve and the transcript is small enough. 
    \begin{align*}
        I(\view_A:\view_B\mid m_1, m_2,\view_E) < \epsilon.
    \end{align*}
    \item Use ``some quantum analogy'' of 1 or 2 above to produce a quantum state $\view'_A:=\rho_{A'}$ from $\rho_E$ such that, 
    \begin{align*}
        \view'_A \view_B \approx_{\poly(\epsilon)} \view_A \view_B. 
    \end{align*}
    \end{itemize}
\end{mdframed}
There are two questions remain to be answered:
\begin{center}
\it Does there exist a query-efficient Eve's strategy that always makes CMI small? \\

Does there exist a quantum analogy of step 1 or 2?
\end{center}

We answer both of the questions affirmatively. 

\paragraph{Decreasing CMI.}
Assume Alice and Bob each makes at most $d$ queries and $H:[2^n] \to \{0,1\}$ be any oracle of domain $[2^n]$ and binary range. We show that, 
\begin{lemma}[Informal] \label{lem:inf}
    For any standard two-round QKA right after Alice receives $m_2$, Eve can run the same Bob $t$ times for some $t \in \{0, 1, \ldots, 2 d n / \epsilon\}$ and store all $t$ copies of Bob's internal states, such that the CMI between Alice and Bob conditioned on Eve's register and $m_1, m_2$ is at most $\epsilon$. 
\end{lemma}

A couple of things we clarify here. First, why don't we set $t := 2 d n / \epsilon$ (the largest value)? This is due to the nature of quantum conditional mutual information. Classically, when we condition on more classical queries, the CMI will never increase. However, quantum entropy and mutual information behave unlike their classical counterparts. Still, we are able to show the existence of such small $t$. The existence of such $t$ will not make our Eve non-uniform, as $t$ is (inefficiently) computable without making any oracle query.\footnote{One can also guess a uniform $t$ if we do not require finding the key with probability close to $1$.}

Second, this strategy works for any oracles. Third, the lemma does not distinguish between whether $m_1$ is classical or quantum. Even if Bob takes a quantum input, as long as we have access to $t$ copies of the same quantum state, we can run Bob with the same pure state $t$ times and thus the lemma still holds. This fact will be useful for the case of quantum public keys.

Finally, the lemma works for any two-round QKA\footnote{{It is a general lemma that works for any two-classical-message quantum interactive protocol. We focus on its application in QKA in this paper.}}. 
As Alice and Bob are asymmetric for general protocols, we do not know how to only simulate Alice's queries and make the CMI small in the general two-round QKA case. We will mention it after finishing the discussion on Merkle-like QKA. 

\paragraph{Sampling Fake Alice.}
When the CMI is small, we need to sample a fake Alice. If the CMI is $0$, Hayden et al.~\cite{hayden2004structure} showed an approach that can be viewed as a quantum analogy of Perspective 1. However, their approach only works for the case of CMI being exactly $0$ and is not robust.

We realize the second interpretation works much better. The work by Fawzi and Renner showed that
\begin{lemma}[Approximate Quantum Markov Chain, \cite{fawzi2015quantum}]\label{thm:qmarkov}
    Let $X, Y, Z$ be three quantum registers, and $\rho_{XYZ}$ be the state. If $I(X:Z|Y) < \epsilon$, then there exists a channel $\mathcal{T}: Y \to Y'Z'$ such that 
    \begin{align*}
        \left| \rho_{XZ} - \sigma_{XZ'} \right|_{\sf Tr} \leq \left| \rho_{XYZ} - \sigma_{XY'Z'} \right|_{\sf Tr} \leq O(\sqrt{\epsilon}),
    \end{align*}
    where $\sigma_{XY'Z'}$ is the state from applying $\mathcal{T}$ on the $Y$ register of $\rho_{XY}$. Furthermore, $\mathcal{T}$ is explicitly (and inefficiently) constructible if knowing the state $\rho_{XYZ}$. 
\end{lemma}

On a high level, \Cref{thm:qmarkov} states that if the CMI $I(X:Z|Y)$ for a tripartite state $\rho_{XYZ}$ is small enough, then we can apply a local channel on $Y$ to generate a state close to the original state in $XY$. This directly provides us a way to sample a fake Alice. In our case, $X$ is  the view of Bob, $Z$ is the view of Alice and $Y$ is the view of Eve and the transcript. Thus, the whole density matrix of $\rho_{XYZ}$ is known by Eve\footnote{$\rho_{XYZ}$ in this case, is not the state under a particular oracle, but the mixed state averaged over the distribution of all oracles. Only in this case, Eve knows $\rho_{XYZ}$.}. By applying the lemma, Eve can sample $Z' := \view'_A$  such that $\view_A' \view_B \approx_{O(\sqrt{\epsilon})} \view_A \view_B$; although $\mathcal{T}$ in this case could be inefficient, it makes no queries. This completes the second step of our Eve. 

\vspace{1em}

Combining the approach of making the CMI arbitrarily small with the quantum Markov chain for sampling a fake Alice, our Eve can attack any Merkle-like QKA in any oracle model. We remark that for a general two-round QKA, \Cref{lem:inf} and \Cref{thm:qmarkov} only cover the part up to receiving the last message $m_2$. In the next subsection, we will discuss how to handle additional queries made after the last message.

\subsection{Ruling out QPKE with Classical-Query Key Generation}

The exact same idea from the previous section also applies to the standard two-round QKA as defined in \Cref{def:informal_QKA}, particularly when $\As_2$ (or the decryption algorithm) makes no queries to the oracle. In this scenario, we let $\view_A$ and $\view_B$ represent the views of Alice and Bob right after $m_1$ is received by Alice, and she has not yet begun working on producing the key. When we sample $\view_A' \view_B \approx \view_A \view_B$, since $\As_2$ only applies a local unitary that is independent of the oracle, Eve can perform the same on the fake view and still successfully obtain the key. However, what if $\As_2$ makes queries?

\medskip

It is not immediately clear whether our CMI-based method works when $\As_1$, $\As_2$, and $\Bs$ all make classical queries only, as in the case of classical KA instead of QKA. 
There are two attempts, that one might immediately come out. 
\begin{itemize}
    \item Attempt 1: run $\As_2$ on the fake $\view_A'$ using any oracle that is compatible with $\view_A'$;
    \item Attempt 2: run $\As_2$ on the fake $\view_A'$ using the real  oracle.
\end{itemize}

Unfortunately, both approaches fail. Consider the following two classical examples for KA. 
\begin{itemize}
    \item Example 1: $\As_1$ does not query and does not send messages, $\Bs$ sends a random $m_2 = x$ and both Alice and Bob agree on $H(x)$. 

    It is easy to see that, even if Eve does not query, $I(\view_A:\view_B|m_2) = 0$ as $\view_A$ is empty. Thus, any oracle $H'$ is consistent with a fake $\view_A'$ but with overwhelming probability, $H'(x) \ne H(x)$. 

    \item Example 2: $\As_1$ queries $H(0)$ and stores it as $y$, but does not send messages to $\Bs$; $\Bs$ sends a random $m_2 = x$ that is not equal to $0$; $\As_2$ queries $H(0)$ again and aborts if $y \ne H(0)$, otherwise the key will be $H(x)$. 

    In an honest execution, Alice and Bob will always agree on a key $H(x)$ for some $x \ne 0$ as $y$ is always equal to $H(0)$. 

    We claim that $I(\view_A:\view_B|m_2) = 0$ as $\view_A$ and $\view_B$ have no intersection queries. Let $\view'_A$ be the fake view that consists of some $y'$. 
    A real oracle with overwhelming probability has $H(x) \ne y'$, which fails the second attempt. 
\end{itemize}

Our solution is to combine both Attempt 1 and Attempt 2.

\paragraph{Solution for Simulating $\As_2$ in the Classical Case.}

Our solution provides an alternative proof for the classical impossibility "PKE is not in Minicrypt" using our CMI-based framework and our solution for simulating $\As_2$ below.

When $\As_1$, $\As_2$, and $\Bs$ are all classical-query algorithms, we propose the following method to execute a fake $\view_A' = (s_A', R_A')$. Here, $s_A'$ represents Alice's random coins, and $R_A'$ represents the query list (a list of input-output pairs).
\begin{itemize}
    \item Run $\As_2$ on the fake $\view_A'$ using the real oracle, except for every $x \in R_A'$, respond with the corresponding image $y \in R_A'$. 
\end{itemize}
In other words, we adjust the real oracle such that it is consistent with $R_A'$; let's denote this modified oracle as $H'$.

Why it works? Assume we sample $\view_A'$ such that $\view_A' \view_B$ has the same distribution as the real views for Alice and Bob. 
The fake view $\view_A'$ together with the real $\view_B$ must be reachable under some oracle (not necessarily the real oracle). Therefore, $R_A'$ and $R_B$ must be consistent. We also know that, since $\view_B$ is the real view of Bob under the real oracle $H$, $R_B$ and $H$ must also be consistent. 

Thus, changing the oracle to be consistent with $R_A'$ will only alter its behavior on those $x \notin R_B$. That means, under oracle $H'$, this KA will still have a strictly positive probability to end up with $\view'_A, \view_B, m_1, m_2$. By perfect completeness, when running $\As_2$ on $\view'_A$ with oracle $H'$, we must obtain the key held by Bob. 

\paragraph{Extending to QPKE with Classical-Query Key Generation.}
For QKA with $\As_1$ making only classical queries (or QPKE with the key generation making classical queries), we can still run our Eve algorithm such that $I(\view_A:\view_B|(m_0,m_1,E)) < \epsilon$. 
However, since $\Bs$ is now quantum, $\view_B$ is some quantum state $\rho_B$ and its query list $R_B$ is no longer defined. 

We repeat our strategy again: sample $\view_A'$ and run $\As_2$ on $\view_A'$ with oracle $H'$ as defined above (mostly the real oracle, but made consistent with $R_A'$). 
If Bob has low query weights on $R_A'$, then our attack still works. If the total query weight of Bob on $R_A'$ is $0$, changing the real oracle to $H'$ will not change $\view_B$ at all. Similarly, if the weight is small, changing the real oracle to $H'$ will only change $\view_B$ by a small amount~\cite{BBBV97}. Thus, we can still argue that $\view_A', \view_B, m_1, m_2$ are reachable under $H'$ and with perfect completeness, we must recover the key. 

What if Bob has a large query weight on some $x \in R_A'$? We imagine a hypothetical Bob, who will first produce $\view_B$, but then keep running itself from the beginning multiple times and randomly measure one of its queries. By doing so, the hypothetical Bob's functionality does not change; the advantage is now the hypothetical Bob has a \emph{classical list} $L_B$ that consists of all its queries with high weights, with a high probability.

If we run our Eve with this hypothetical Bob, $\view_A'$ must have $R_A'$ consistent with $L_B$. Thus, changing the oracle to be consistent with $R_A'$ will  only change its behavior on those $x \not\in L_B$, or in other words, those $x$ that do not have a large query weight! Then we can use the argument in the previous paragraphs and claim that $\view_A', \view_B:=(\rho_B, L_B), m_1, m_2$ are reachable under $H'$ and with perfect completeness, we can get the key. 

\paragraph{Handling Quantum Public Keys}

Quantum public keys (or quantum $m_1$) are handled without additional efforts using our CMI-based Eve. This further demonstrates the versatility and power of our new framework; previous approaches based on classical-query Eve do not extend to this case. 

When $m_1$ is quantum, the challenge arises in how Eve can effectively run multiple copies of Bob on quantum input $m_1$. In the context of QPKE, the attacker needs access to multiple copies of the quantum $m_1$. If $m_1$ is a pure state, it becomes evident that running the same Bob on the same pure state $m_1$ is feasible, allowing us to minimize the CMI accordingly. Alternatively, if a query-efficient perfect cloner for $m_1$ exists, multiple runs of the same Bob become possible.

Given that the subsequent analysis relies solely on the capability to run Bob multiple times on input $m_1$, we can extend our conclusion to rule out perfect complete QPKE with quantum public keys and classical-query key generation under the condition that the quantum public keys are either pure or query-efficiently clonable.

\section{Preliminaries}

We refer reader to \cite{nielsen2002quantum} for more details about quantum computing and quantum information. Below, we mention some backgrounds that are heavily used in this work. 

\subsection{Distance measures}
Let us recall the definition of total variation distance and trace distance.

\begin{definition}[Total variation distance]
    For two probabilistic distributions $D_X, D_Y$ over the same finite domain $\mathcal{X}$, we define its total variation distance as
    \begin{align*}
        TV(D_X,D_Y)=\frac{1}{2}\sum_{x\in\mathcal{X}}|D_X(x)-D_Y(x)|.
    \end{align*}
\end{definition}

\begin{definition}[Trace distance]
    For two quantum states $\rho,\sigma$, the trace distance between the two states is
    \begin{align*}
        TD(\rho,\sigma)=\frac{1}{2}\Tr\left[\sqrt{(\rho-\sigma)^\dagger(\rho-\sigma)}\right]=\sup_{0\leq\Lambda\leq I}\Tr[\Lambda(\rho-\sigma)].
    \end{align*}
\end{definition}
\subsection{Quantum Oracle Model and Random Oracle}
A quantum oracle algorithm equipped with access to $H\colon [2^{n_{\lambda}}]\to [2^{m_{\lambda}}]$ is expressed as a series of unitaries: $U_1$, $U_H$, $U_2$, $U_H$, $\cdots$, $U_T$, $U_H$, $U_{T+1}$. Here, $U_i$ denotes a local unitary acting on the algorithm's internal register. Oracle access to $H$ is defined by a unitary transformation $U_H$, where $\ket{x, y}$ is transformed to $\ket{x, y + H(x)}$. For an oracle algorithm $\As$, we would use $\As^H$ to denote the algorithm $\As$ has classical access to the oracle $H$, and $\As^{\ket{H}}$ to denote $\As$ has quantum access to the oracle.

We would also consider the case when the oracle $H\colon [2^{n_{\lambda}}]\to\{0,1\}$ is sampled from some distribution of oracles $\Hs_{\lambda}$. We would call some primitive in the quantum random oracle model(QROM) if the distribution $\Hs_{\lambda}$ is uniformly random over all possible oracles.

\subsection{Entropy}
\begin{definition}[Von Neumann Entropy]
    Let $\rho \in \mathbb{C}^{2^n}$ be a quantum state describing a system $\A$. Let $\ket {\phi_1}, \ket {\phi_2}, \cdots, \ket {\phi_{2^n}}$ be the eigenbasis of $\rho$; $\rho$ is written in this eigenbasis as $\sum_i \eta_i \ket {\phi_i} \bra {\phi_i}$. 
    
    Then its Von Neumann Entropy is denoted by $S(\rho)$ (or $S(\A)_\rho$), 
    \begin{align*}
        S(\A)_\rho = S(\rho) = - \sum_i \eta_i \log(\eta_i). 
    \end{align*}

    Given a composite quantum system $\A\B$ having joint state $\rho_{\A\B}$, we define the conditional Von Neumann Entropy as $S(\A|\B)_\rho$, 
    \begin{align*}
        S(\A|\B)_\rho = S(\A\B)_\rho - S(\B)_\rho.
    \end{align*}
\end{definition}

\noindent Below, we often omit $\rho$ in the definition when the quantum state is clear in the context. For example, $S(\A)$ and $I(\A : \B)$ instead of $S(\A)_\rho$ and $I(\A:\B)_\rho$. 

\begin{fact}[\cite{nielsen2002quantum}]
    \label{fact:qc_entropy}
    Suppose $p_k$ are probabilities, $\ket{k}$ are orthogonal basis of a system $\A$ and $\rho_k$ are quantum states for another system $\B$. Then
    \[
      S\left(\sum_{k} p_k\ket{k}\bra{k}\otimes \rho_k\right)=H(p_k)+\sum_k
      p_k S(\rho_k)
    \]
    where $H(p_k)$ is the Shannon entropy of distribution
    $p_k$.
\end{fact}

\begin{fact}[\cite{wilde2011classical}] \label{fact:separable_basic}
    If $\rho_{\A\B}$ is a separable state, then $S(\A|\B)\geq 0$.
\end{fact}

\begin{definition}[Mutual Information]
    Let $\rho$ be a quantum state describing two joint systems $\A$ and $\B$. Then the mutual information between the system $A$ and $B$ is denoted by $I(\A:\B)$, 
    \begin{align*}
        I(\A:\B) = S(\A) + S(\B) - S(\A\B).%
    \end{align*}
\end{definition}

\begin{definition}[Conditional Mutual Information]
    Let $\rho$ be a quantum state describing three joint systems $\A$, $\B$ and $\C$. Then the conditional mutual information $I(\A:\B|\C)$, 
    \begin{align*}
        I(\A:\B|\C) &= S(\A\C) + S(\B\C) - S(\A\B\C) - S(\C). %
    \end{align*}
\end{definition}

\begin{fact}[Chain rule]\label{fact:chain_rule}
$I(\A_1,\A_2,\cdots,\A_t:\B\mid \C)=\sum_{i=1}^t I(\A_i:\B\mid \C,\A_1,\cdots,\A_{i-1})$.
\end{fact}
\begin{fact}
    Let $\A\B\C$ be a composite quantum system. When a unitary is applied on $\A$, it will not change $S(\A)$. Similarly, the local unitary on $\A$ will not change $I(\A:\B)$ or $I(\A:\B|\C)$. 
\end{fact}
\begin{proof}
    This directly follows from the definition of $S(\A)$ and applying any unitary will not change the spectrum of a density matrix. 
\end{proof}

The strong subadditivity for (conditional) mutual information concludes that both mutual information and conditional mutual information are always non-negative. 
\begin{lemma}[Strong Subadditivity, \cite{araki1970entropy}]
\label{lem:strong_subadditivity}
    Given Hilbert spaces $\A, \B, \C$, 
    \begin{align*} 
    S(\A\C) + S(\A\B) \geq S(\A\B\C) + S(\C). 
    \end{align*}

    (The conditional form) Given Hilbert spaces $\A, \B, \C, \D$, 
    \begin{align*} 
    S(\A\C|\D) + S(\A\B|\D) \geq S(\A\B\C|\D) + S(\C|\D). 
    \end{align*}
\end{lemma}

\subsection{Operational Meaning of Conditional Mutual Information: Approximate Quantum Markov Chain}

Fawzi and Renner \cite{fawzi2015quantum} provided a nice  characterization of quantum states for which the
conditional mutual information is approximately zero. Intuitively, if $I(\A:\B\mid \E)$ is small, then $\B$ can be approximately reconstructed from $\E$. 
\begin{theorem}[\cite{fawzi2015quantum}, restate of Lemma \ref{thm:qmarkov}]\label{thm:quantum_mutual_info_operational}
For any state $\rho_{{\sf AEB}}$ over systems $\A\E\B$, there exists a channel ${\cal T} : {\sf E} \to {\sf E} \otimes {\sf B}'$ such that the trace distance between the reconstructed state $\sigma_{\sf A'E'B'} = {\cal T}(\rho_{\A\E})$ and the original state $\rho_{\sf AEB}$ is at most 
\begin{align*}
    \sqrt{\ln 2 \cdot I(\A:\B|\E)_\rho}.
\end{align*}
\end{theorem}

\subsection{Quantum Key Agreement and Quantum Public Key Encryption}\label{sec:key_def}
In the following, we provide  formal definitions of quantum key agreement (QKA)
and quantum public key encryption (QPKE) in the oracle model.

\begin{definition}[Quantum Key Agreement in the Oracle Model]
Let $\lambda \in \mathbb{Z}_+$ be the security parameter and $\Hs_\lambda$ be a distribution of oracles. 
Let $H \gets \Hs_\lambda$ be a classical oracle, drawn according to $\Hs_{\lambda}$.

A key agreement protocol consists of two parties Alice and Bob, who start with all-zero states and have the ability to apply any quantum operator, get quantum access to $H$, and send classical messages to each other.

Both Alice and Bob can make at most $\poly(\lambda)$ number of quantum queries to $H$. Finally, Alice and Bob output classical strings $k_A, k_B$\footnote{$k_A, k_B$ can be of any length (even exponential in $\lambda$). Our impossibility results apply to protocols with any output key length.}. We would call the set of classical messages between Alice and Bob, denoted by $\Pi$, the transcript of the key exchange protocol.
\end{definition}

A key agreement protocol should satisfy both correctness and security.

\begin{definition}[Correctness]
    Let $k_A, k_B$ be the keys outputted in the protocol. Then $\Pr[k_A = k_B] \geq 1/q(\lambda)$ for some polynomial $q(\cdot)$, where the probability is taken over the randomness of Alice and Bob's channels, and the random choice of oracle $H$.  
\end{definition}

\begin{definition}[Security]
    For any eavesdropper Eve that makes at most $\poly(\lambda)$ number of quantum queries to $H \gets \Hs_\lambda$, eavesdrops classical communication between Alice and Bob and outputs $k_E$, the probability that $\Pr[k_A = k_E]$ is negligible in $\lambda$. 
\end{definition}

We are interested in whether there exists a protocol that satisfies both correctness and security in the QROM. 
It is worth noting that the ability of making quantum queries is essential in our setting; when Alice and Bob can only make polynomially many classical queries, secure key agreement does not exist; i.e., there always exists an eavesdropper making polynomially many classical queries and breaking it~\cite{impagliazzo1989limits,barak2009merkle}.

We also focus on breaking quantum public key encryption schemes in the oracle model.  

\begin{definition}[Quantum Public Key Encryption in the Oracle Model]\label{def:QKA}
Let $\lambda \in \mathbb{Z}_+$ be the security parameter and $\Hs_\lambda$ be a distribution over oracles $H\colon [2^{n_\lambda}]\to\{0,1\}$. 
Let $H \gets \Hs_\lambda$ be a classical oracle, 
a QPKE scheme in the oracle model consists of three algorithms ($\gen, \enc, \dec$), each of which is allowed to make at most $d(\lambda)=\poly(\lambda)$ quantum queries to $H$:

\begin{itemize}
    \item $\gen^{\ket{H}}(1^\lambda)\to(\pk,\sk)$: The quantum key generation algorithm that generates a pair of classical public key $\pk$ and secret key $\sk$.
    \item $\enc^{\ket{H}}(\pk,m)\to \ct$: the quantum encryption algorithm that takes a public key $\pk$, the plaintext $m$, produces the ciphertext $\ct$. 
    \item $\dec^{\ket{H}}(\sk,\ct)\to m'$: the quantum decryption algorithm that takes secret key $\sk$ and ciphertext $\ct$ and outputs the plaintext $m'$. 
\end{itemize}

The algorithms should satisfy the following requirements:
\begin{description}
    \item[Completeness] $\Pr\left[\dec^{\ket{H}}\left(\sk,\enc^{\ket{H}}(\pk,m)\right)=m\colon \gen^{\ket{H}}(1^\lambda)\to(\pk,\sk)\right]\geq 1-\negl(\lambda)$.
    \item[IND-CPA Security] For any adversary $\mathcal{E}^{\ket{H}}$ that makes $\poly(\lambda)$ queries, for every two plaintexts $ m_0\neq m_1$ chosen by $\Es^{\ket{H}}(\pk)$, we have
    \begin{align*}
        \Pr\left[\mathcal{E}^{\ket{H}}\left(\pk,\enc^{\ket{H}}(\pk,m_b)\right)=b\right]\leq \frac{1}{2}+\negl(\lambda).
    \end{align*}
    
\end{description}
\end{definition}

Especially in~\Cref{sec:classical_keygen}, we will focus on when the key generation algorithm $\gen^H$ is an algorithm with classical access to the oracle $H$.

It is a folklore result that we can construct a two-round key agreement protocol from a public key encryption scheme as follows:

\begin{enumerate}
    \item Alice runs $\gen^{\ket{H}}$ to produce a public key $\pk$ and a secret $\sk$, and sends the public key $\pk$ as the first message $m_1$ to Bob. We call Alice's algorithm at this stage $\As_1$. 
    \item Bob upon receiving $\pk$, it samples a uniformly random classical string $k$ as the key. It computes the ciphertext $\ct \gets \enc^{\ket{H}}(\pk, k)$ and sends $m_2 := \ct$ to Alice. We call Bob's algorithm at this stage $\Bs$. 
    \item Alice will then run $\dec^{\ket{H}}(\sk, \ct)$ to output $k'$, which is her guess of $k$. We call Alice's algorithm at this stage $\As_2$. 
\end{enumerate}
We further notice that the only information required by $\As_2$ from $\As_1$ is the secret key $\sk$. Thus if $\sk$ is classical, we can assume without loss of generality that the internal state of $\As_1$ at its termination is a mixed state in the computational basis.

We would also consider the recently proposed QPKE with quantum public key~\cite{barooti2023publickey,kitagawa2023quantum,coladangelo2023quantum} in~\Cref{sec:classical_keygen}. We would only focus on the variant where the protocol is perfect complete, and key generation algorithms can only make classical queries.

\begin{definition}[QPKE with quantum public key]\label{def:PKE_qpk}
    Let $\lambda \in \mathbb{Z}_+$ be the security parameter and $\Hs_\lambda$ be a distribution over oracles $H\colon [2^{n_\lambda}]\to\{0,1\}$. 
Let $H \gets \Hs_\lambda$ be a classical oracle, 
a QPKE scheme in the oracle model consists of four algorithms ($\skgen,\pkgen, \enc, \dec$), each of which is allowed to make at most $d(\lambda)=\poly(\lambda)$ queries to $H$:

\begin{itemize}
    \item $\skgen^{{H}}(1^\lambda)\to\sk$: The secret key generation algorithm that generates a classical secret key $\sk$. 
    \item $\pkgen^{{H}}(\sk)\to\rho_{\pk}$: The public key generation algorithm that takes the secret key $\sk$ and generates a quantum public key $\rho_{\pk}$.
    \item $\enc^{\ket{H}}(\rho_{\pk},m)\to \ct$: the quantum encryption algorithm that takes a public key $\pk$, the plaintext $m$, produces the (possibly quantum) ciphertext $\rho_{\ct}$. 
    \item $\dec^{\ket{H}}(\sk,\rho_{\ct})\to m'$: the quantum decryption algorithm that takes secret key $\sk$ and ciphertext $\ct$ and outputs the plaintext $m'$. 
\end{itemize}

The algorithms should satisfy the following requirements:
\begin{description}
    \item[Perfect Completeness] $$\Pr\left[\dec^{\ket{H}}\left(\sk,\enc^{\ket{H}}(\rho_{\pk},m)\right)=m\colon \skgen^{{H}}(1^\lambda)\to\sk,\pkgen^H(\sk)\to\rho_{\pk}\right]=1.$$
    \item[IND-CPA Security] For any adversary $\mathcal{E}^{\ket{H}}$ that makes $\poly(\lambda)$ queries, given any polynomial copies of public key $\rho_{\pk}^{\otimes t(\lambda)}$, for every two plaintexts $ m_0\neq m_1$ chosen by $\Es(\rho_{\pk}^{\otimes t(\lambda)})$, we have
    \begin{align*}
        \Pr\left[\mathcal{E}^{\ket{H}}\left(\rho_{\pk}^{\otimes t(\lambda)},\enc^{\ket{H}}(\rho_{\pk},m_b)\right)=b\right]\leq \frac{1}{2}+\negl(\lambda).
    \end{align*}
    
\end{description}
\end{definition}

We call the public key pure if given $\sk$, $\rho_{\pk}\leftarrow \pkgen(\sk)$ is a pure state, and call the public key clonable if there is some polynomial query algorithm $\Ds^{\ket{H}}$ that takes $\rho_{\pk}=\sum p_i\ket{\psi_i}\bra{\psi_i}$ in its eigenvector decomposition, and generates the state $\rho'=\sum  p_i\ket{\psi_i}\bra{\psi_i}^{\otimes{t(\lambda)}}$ for some polynomial $t(\cdot)$.

\begin{remark}
The definition of clonable here might seem odd at first glance, we would give some examples here for further explanations. The public key generation algorithm from~\cite{kitagawa2023quantum,cryptoeprint:2023/500} $\pkgen(\sk)$ generates a $\rho_{\pk}=(\pk_r,\ket{\psi_r})$ according to some private coin $r$ of $\pkgen(\sk)$, where $\pk_r$ is a classical string, $\ket{\psi_r}$ is a pure quantum state. Thus our cloning algorithm $\Ds$ can be seen as given a sample $(\pk_r,\ket{\psi_r})$, it can generate multiple copies of state $\ket{\psi_r}$.
\end{remark}

\section{Helper Lemmas}\label{sec:helper}
In this section, we introduce three helper lemmas. \Cref{lemma:permuation-invariance} tells how to decrease CMI. \Cref{lem:cmi_op} and \Cref{lem:ccnotcmi} claim that classical communication does not increase CMI. \Cref{lem:support} will be used in \Cref{sec:classical_keygen}.

\subsection{Repetition Decreases CMI}
\begin{definition}[Permutation Invariance]
Let $\A_1,\A_2,\A_3,\ldots,\A_t,\B$ be $(t+1)$-partite quantum system. Given the joint state $\rho_{\B\A_1\A_2\cdots \A_t}$, 
we say $A_1,\ldots,A_t$ are permutation invariant, if for any permutation $\pi$ on $[t]$, we have 
 \[
 \rho_{\B\A_1\A_2\cdots \A_t}=\rho_{\B\A_{\pi(1)}\A_{\pi(2)}\cdots \A_{\pi(t)}}.
 \]
\end{definition}
\begin{lemma}\label{lemma:permuation-invariance}
Let $\A_1,\A_2,\A_3,\ldots,\A_t,\B,\C$ be $(t+2)$-partite quantum system. Suppose the state of the composite system $\rho_{\B\C\A_1\A_2\cdots \A_t}$ is fully separable. If $\A_1,\A_2,\A_3,\ldots,\A_t$ are permutation invariant, %
 then there is a $0\leq i\leq t-1$ such that 
 \[
 I(\A_t: \B\mid \C,\A_1,\ldots,\A_{i})_{\rho}\leq S(\B)/t.
 \]
\end{lemma}
\begin{proof} By the chain rule of conditional mutual information (see Fact \ref{fact:chain_rule}), we have
\begin{equation}\label{eq:non-interactive-eq1}
\sum_{i=1}^t I(\A_i:\B\mid \C,\A_1,\ldots,\A_{i-1})=I(\A_1,\ldots,\A_t:\B\mid \C)
\end{equation}
Besides,
\begin{equation}\label{eq:non-interactive-eq2}
I(\A_1,\ldots,\A_t:\B\mid \C)=S(\B\mid \C)-S(\B\mid \C,\A_1,\ldots,\A_t)\leq S(\B\mid \C)\leq S(\B),
\end{equation}
where the inequalities are by \Cref{fact:separable_basic} and that $I(\B:\C)=S(\B)-S(\B\mid \C)$ is always non-negative. By \eqref{eq:non-interactive-eq1} and \eqref{eq:non-interactive-eq2}, 
there must exist $i\in [t]$ such that $I(\A_i:\B\mid \C,\A_1,\ldots,\A_{i-1})\leq S(\B)/t$.
Finally, by the permutation invariance, we have $I(\A_i:\B\mid \C,\A_1,\ldots,\A_{i-1})=I(\A_t:\B\mid \C,\A_1,\ldots,\A_{i-1})$. Now we finish the proof.
\end{proof}

\subsection{Classical Communication does not Increase CMI}
\begin{lemma}[Local quantum operation does not increase CMI]
    \label{lem:cmi_op}
    Let $\A\B\C$ be a composite quantum system. After performing a quantum operation $M$ on $\A$, the state of the system becomes $\A'\B'\C'$. Then $I(\A':\B'|\C')\leq I(\A:\B|\C)$. 
\end{lemma}

\begin{proof}
    We introduce another quantum system $\D$, initialized as zero. The quantum operation $M$ can be treated as first performing a unitary transformation $U$ on $\A\D$ and then discarding $\D$.
    \begin{align*}
        I(\A:\B|\C) &= I(\A\D:\B|\C) \\
        &= I(U(\A\D)U^\dagger:\B|\C) \\
        &= I(\A'\D':\B'|\C') \\
        &= S(\A'\D'|\C')-S(\A'\D'\B'|\C')+S(\B'|\C').
    \end{align*}
    
    By the conditional form of strong sub-additivity (see the condition form of \Cref{lem:strong_subadditivity}), $S(\A'\D'|\C')-S(\A'\D'\B'|\C')\geq S(\A'|\C')-S(\A'\B'|\C')$. Then
    \begin{align*}
        I(\A:\B|\C) &\geq S(\A'|\C')-S(\A'\B'|\C')+S(\B'|\C') \\
        &= I(\A':\B'|\C').  \qedhere
    \end{align*}
\end{proof}

\begin{lemma}[Sending classical message does not increase CMI] \label{lem:ccnotcmi} 
    Let $\A\B\C$ be a composite quantum system and $\A=(\mathsf{W}_\A, \mathsf{M}_\A)$ 
    where $\mathsf{W}_\A$ is the working register and $\mathsf{M}_\A$ is the message register containing a classical state. After 
    $\B$ and $\C$ both obtain a copy of $\mathsf{M}_\A$, the system becomes $\A'\B'\C'$.
    Then $I(\A':\B'|\C')\leq I(\A:\B|\C)$.
\end{lemma}

The following claim will be used.
\begin{claim}[Copying classical state does not change entropy]
    \label{claim:copy_classical}
    Given a system $\P\Q\R$ where $\Q$ and $\R$ contain identical classical states.    
    Then $S(\P\Q\R)=S(\P\R)=S(\P\Q)$.
\end{claim}

\begin{proof}
    The joint state of $\P\Q\R$ can be written as  
    \[\sum_k p_k\rho_{\P}^{(k)}\otimes \ket{k}\bra{k}_{\Q}\otimes \ket{k}\bra{k}_{\R}\]
    where $p_k$ are probabilities,  $\rho_{\P}^{(k)}$ are quantum states in $\P$, and $\ket{k}$ are computational basis of $\Q$ (also $\R$). 
    Then by \Cref{fact:qc_entropy}, 
    \begin{align*}
        S(\P\Q\R)
        &= S\left(\sum_k p_k\rho_{\P}^{(k)}\otimes \ket{k}\bra{k}_\Q\otimes \ket{k}\bra{k}_{\R}\right) \\
        &= H(p_k)+\sum_k p_kS\left(\rho_{\P}^{(k)}\right) 
        = S\left(\sum_k p_k\rho_{\P}^{(k)}\otimes \ket{k}\bra{k}_\Q\right)
        = S(\P\Q).
    \end{align*}
    By symmetry, $S(\P\Q\R)=S(\P\R)=S(\P\Q)$.
\end{proof}

\begin{proof}[Proof of \Cref{lem:ccnotcmi}]
    Let $\B'=(\B, \Msf_\B), \C'=(\C, \Msf_\C)$ where $\Msf_\B, \Msf_\C$ are the message registers which contain a copy of $\mathsf{M}_\A$. 
    Then 
    \begin{align*}
        I(\A':\B'|\C') 
        &= I(\Wsf_\A,\Msf_\A:\B,\Msf_\B|\C,\Msf_\C) \\
        &= S(\Wsf_\A, \C, \Msf_\A, \Msf_\C)
        + S(\B, \C, \Msf_\B, \Msf_\C)
        - S(\Wsf_\A, \B, \C, \Msf_\A, \Msf_\B, \Msf_\C)
        - S(\C, \Msf_\C).
    \end{align*}
    Since $\Msf_\A, \Msf_\B, \Msf_\C$ contain identical classical states, 
    by \Cref{claim:copy_classical},
    we can remove $\Msf_\A,\Msf_\B$ from above equation while keeping each 
    term unchanged. Then
    \begin{align*}
        I(\A':\B'|\C') 
        &= S(\Wsf_\A, \C, \Msf_\C)
        + S(\B, \C, \Msf_\C)
        - S(\Wsf_\A, \B, \C, \Msf_\C)
        - S(\C, \Msf_\C) \\
        &= I(\Wsf_\A:\B|\C, \Msf_\C) \\
        &= I(\Wsf_\A,\Msf_\C:\B|\C) - I(\Msf_\C:\B|\C). 
    \end{align*}
    By the non-negativity of $I(\Msf_\C:\B|\C)$, 
    \begin{align*}
        I(\A':\B'|\C') 
        &\leq I(\Wsf_\A,\Msf_\C:\B|\C) \\
        &=S(\Wsf_\A,\C,\Msf_\C) + S(\B,\C)
        -S(\Wsf_\A,\B,\C,\Msf_\C) - S(\C).
    \end{align*}
    Since $\Msf_\A$ and $\Msf_\C$ are identical classical states, $\Msf_\C$ can be replaced with $\Msf_\A$ by \Cref{claim:copy_classical}.
    Then
    \begin{align*}
        I(\A':\B'|\C') 
        &\leq S(\Wsf_\A,\C,\Msf_\A) + S(\B,\C)
        -S(\Wsf_\A,\B,\C,\Msf_\A) - S(\C) \\
        &=I(\Wsf_\A,\Msf_\A:\B|\C)=I(\A:\B|\C). \qedhere
    \end{align*}

\end{proof}

\subsection{Other Useful Lemmas}
We would also need the following lemma.
\begin{lemma}\label{lem:support}
    For two classical probabilistic distributions $D_X$ and $D_Y$ over the same domain, if\\ $TV(D_X,D_Y)\leq \epsilon$, we have that
    \begin{align*}
        \Pr_{x\leftarrow D_X}[x\notin\supp(D_Y) ]\leq 2\epsilon.
    \end{align*}
\end{lemma}

\begin{proof}
    We use $p_x^X,p_x^Y$ to denote the probability of $x$ drawn from $ D_X,D_Y$ respectively.
    \begin{align*}
        \sum_{x\notin\supp(D_Y) }p_x^X\leq\sum_{x}|p_x^X-p_x^Y|=2TV(D_X,D_Y)\leq2\epsilon.\quad\quad \qedhere
    \end{align*}
\end{proof}

\section{Non-Interactive Quantum Key Agreement}\label{sec:noninteractive}

 We say a QKA is \emph{non-interactive} if all queries are made before communication. Formally, 

\begin{definition}[Non-interactive Quantum Key Agreement]
\label{def:NIKA}
    A non-interactive key agreement protocol between Alice and Bob consists of the following steps:
    \begin{enumerate}
        \item Let $\lambda > 0$ be a security parameter and $H \gets \Hs_\lambda$. 
        \item Alice and Bob each makes $d_{\lambda}=\poly(\lambda)$ queries to $H$. 
        \item Alice and Bob continue an arbitrary number of rounds of classical communication and local quantum operations, but will never make queries to $H$. 
        \item  Eventually, Alice and Bob will output $k_A, k_B$. 
    \end{enumerate}
\end{definition}
\begin{theorem}\label{thm:non-interactive}
Non-interactive QKA  does not exist in any oracle model.
\end{theorem}

\begin{proof}
Let $e_{\lambda}$ denote the number of qubits on which the quantum query unitary $U_H$ acts, i.e., the total length of the input register and the output register. For example, if $H:[2^{n_\lambda}]\rightarrow \{0,1\}$ is a random function, then $e_{\lambda}=n_{\lambda}+1$.
Let $\rho_{\A}^0$ and $\rho_{\B}^0$ denote the states of Alice and Bob respectively right after the query algorithm but before any communication.

First, Eve repeatedly runs the same Alice's query algorithm $t=C d_{\lambda}e_{\lambda}$ times ($C$ is some polynomial determined afterwards) and prepares the state $\rho_{\E}^0$, which consists of $t$ registers $\A_1^0,\ldots,\A_t^0$ of $\rho_{\A}^0$'s copies. Observing that $\A^0,\A_1^0,\ldots,\A_t^0$ are permutation invariant, by \Cref{lemma:permuation-invariance}, we have
\[
I(\A^0:\B^0\mid \E^0)\leq \frac{S(\B^0)}{t+1}\leq \frac{2d_{\lambda}e_\lambda}{t+1}\leq \frac{1}{2C}.
\]
The second inequality is because we can implement the query unitary $U_H$ by a quantum communication process: suppose there are two parties, namely Bob and Oracle; if Bob wants to apply the unitary $U_H$, then
\begin{enumerate}
\item Bob sends its input register and output register, $e_\lambda$ qubits in total, to Oracle;
\item Oracle applies $U_H$ on these $e_\lambda$ qubits and then sends them back to Bob.
\end{enumerate}
By subadditivity of entropy, the entropy of Bob can increase by at most $2e_{\lambda}$ bits through the above quantum communication process.
Since $\B^0$ is prepared from a pure state by making $d$ quantum queries, we have that $S(\B^0)\leq 2d_\lambda e_{\lambda}$.

Let $\rho_{A}^f$ and $\rho_B^f$ denote the states of Alice and Bob respectively right after finishing the communication but before outputting the key. Since classical communication does not increase CMI (see \Cref{lem:ccnotcmi}), we have
\[
I(\A:\B\mid \E,\Pi)_{\rho_{ABE}^f}\leq I(\A:\B\mid \E)_{\rho_{ABE}^0}\leq \frac{1}{2C}.
\]

Finally, Eve applies the channel in \Cref{thm:qmarkov} and obtains a fake view $\hat{\A}^f$ of $\A^f$ such that the joint state of $\rho_{\hat{A}B}^f$ is $O(1/\sqrt{C})$-close to $\rho_{{A}B}^f$. By letting $C=O(1/\epsilon^2)$ , Eve can use $\hat{\A}^f$ to generate $k_E$ such that $\Pr[k_B=k_E]$ is $\epsilon$-close to $\Pr[k_A=k_E]$.
\end{proof}
\begin{remark}\label{rem:prs} (a). \Cref{thm:non-interactive} directly implies the separation of pseudorandom quantum states (PRS) and QKA. This is because PRS can be prepared by random queries, and without loss of generality, we can assume all random queries are made at the very beginning before any communication.

(b). In the proof of \Cref{thm:non-interactive}, Eve makes $O(d_{\lambda}^2 e_\lambda)$ quantum queries. This bound can be improved to $\widetilde{O}(d_\lambda^2)$ if Alice and Bob are restricted to making non-adaptive queries (i.e., queries can depend on neither
any previous queries nor classical communication) to a random function. The famous Merkle Puzzles and PRS-based protocols are both examples of this non-adaptive case. We refer interested readers to \Cref{sec:non-adaptive} for details.

(c). By almost the same proof, \Cref{thm:non-interactive} can be extended to include quantum oracles, e.g., Haar random oracles, where the query unitary operator is chosen from some Haar measure.
\end{remark}

We can further notice that the non-interactive requirement can be relaxed to one-sided if we are considering two-round key agreement protocols. That is, Bob can make queries after receiving the first message from Alice, while Alice can only make queries before communication. Viewing the protocol in the model of QPKE, we have the following theorem:

\begin{theorem}[Restate of \Cref{thm:main_theorem3}]
    For any QPKE scheme with a classical public key and a classical/quantum ciphertext in any oracle model, if it satisfies the following conditions:
    \begin{enumerate}
        \item The key generation algorithm makes at most $d$ quantum queries;
        \item The encryption algorithm makes at most $d$ quantum queries;
        \item The decryption algorithm makes no queries.
    \end{enumerate}
    There exists an adversary Eve that could break the public key encryption scheme with probability $1-O(\epsilon)$ by making $O(d^2e/\epsilon^2)$ number of oracle queries.
\end{theorem}

\begin{proof}
    In this case, if Eve generates $O(de/\epsilon^2)$ copies of message $m_1$ and runs Bob $\Bs(m_1)$ according to the first message $m_1$, we can see that the permutation invariant condition still holds. Consider the stage before Bob sends the second message, similar to~\Cref{thm:non-interactive}, we have that
    \begin{align*}
        I(\A:\B\mid \E)\leq \frac{1}{2C}\leq O(\epsilon^2).
    \end{align*}
    Thus by applying the channel in~\Cref{thm:quantum_mutual_info_operational} to $\E$, we can obtain a state $\rho_{A'B}$ that is $O(\epsilon)$ close to the state $\rho_{AB}$. Now Eve receives the message register from Bob and runs the decryption algorithm on $\A'$, by the completeness of the protocol we can obtain that $$\Pr[k_B=k_E]\geq \Pr[k_A=k_B]-O(\epsilon)\geq 1-O(\epsilon)-\negl(\lambda).$$
\end{proof}

Note that our impossibility result can be extended to the Classical Communication One Quantum Message Key Agreement (CC1QM-KA) model from~\cite{grilo2023towards}. In their security definition, the quantum message channel is unauthenticated, thus it can be modified by the adversary $\Es$. As long as the completeness of the original key agreement protocol is $1-\negl(\lambda)$, we can apply the gentle measurement lemma to ensure that $\Pr[k_A=k_B=k_E]\geq 1-O(\sqrt{\epsilon})$.

\section{Public Key Encryption with a Classical Key Generation}\label{sec:classical_keygen}

In this section, we devote ourselves to proving black-box separation results for public key encryption schemes with a classical key generation process and one-way functions.

Recall the definitions of QPKE schemes from~\Cref{sec:key_def}. We will focus on the specific public encryption schemes defined as follows:

\begin{definition} \label{def:PKEck}
A public key encryption scheme with a classical key generation process, relative to a random oracle $H\leftarrow \Hs_{\lambda}$ consists of the following three bounded-query quantum algorithms:
\begin{itemize}
    \item $\gen^H(1^\lambda)\to(\pk,\sk)$: The key generation algorithm that generates a pair of public key $\pk$ and secret key $\sk$.
    \item $\enc^{\ket{H}}(\pk,m)\to \ct$: the encryption algorithm that takes a public key $\pk$, the plaintext $m$, and the randomness $r$, produces the ciphertext $\ct$. 
    \item $\dec^{\ket{H}}(\sk,\ct)\to m'$: the decryption algorithm that takes secret key $\sk$ and ciphertext $\ct$ and outputs the plaintext $m'$. 
\end{itemize}

Here we use $H$ to denote the algorithm has classical access to the oracle, and $\ket{H}$ to denote quantum access to the oracle.

The algorithms should satisfy the following requirements:
\begin{description}
    \item[Perfect Completeness] $\Pr\left[\dec^{\ket{H}}\left(\sk,\enc^{\ket{H}}(\pk,m)\right)=m\colon \gen^H(1^\lambda)\to(\pk,\sk)\right]=1$.
    \item[IND-CPA Security] For any QPT adversary $\mathcal{E}^{\ket{H}}$, for every two plaintexts $ m_0\neq m_1$ chosen by $\Es^{\ket{H}}(\pk)$ we have
    \begin{align*}
        \Pr\left[\mathcal{E}^{\ket{H}}\left(\pk,\enc^{\ket{H}}(\pk,m_b)\right)=b\right]\leq \frac{1}{2}+\negl(\lambda).
    \end{align*}
    
\end{description}

\end{definition}

To prove a black box separation, we would view the public key encryption scheme as a two-round key agreement protocol as described before, and utilize the previous tools from information theory. From the perspective of a key agreement protocol, it could be viewed as a two-stage Alice: before sending the first message, Alice would first make classical queries to the oracle $H$; And after receiving the message from Bob, it would make quantum queries to the oracle, and output the key they agree on. We denote the first stage as $\As_1$ and the second stage as $\As_2$.

We would use the following lemma from~\cite{BBBV97}.

\begin{lemma}\label{lem:BBBV}
Consider a quantum algorithm $\Bs$ that makes $d$ queries to an oracle $H$. Denote the quantum state immediately after $t$ queries to the oracle as 
\begin{align*}
    \ket{\psi_t}=\sum_{x,w}\alpha_{x,w,t}\ket{x,w},
\end{align*}
where $w$ is the content of the workspace register. Denote the query weight $q_x$ of input $x$ as 
\begin{align*}
    q_x=\sum_{t=1}^d\sum_{w}|\alpha_{x,w,t}|^2.
\end{align*}
For any oracle $\Tilde{H}$, denote $\ket{\phi_d}$ as the final state before measurement obtained by running $\Bs$ with oracle $\Tilde{H}$, we have that
\begin{align*}
    \lVert\ket{\psi_d}-\ket{\phi_d}\rVert\leq 2\sqrt{d}\sqrt{\sum_{x\colon\Tilde{H}(x)\neq H(x)}q_x}.
\end{align*}
\end{lemma}

To define Eve's algorithm, we would first introduce the following algorithm $\Bs'$ as the modified algorithm of $\Bs$:
\begin{enumerate}
    \item On receiving message $m_1$, first run $\Bs(m_1)$, and stop before the final output measurement.
    \item Repeat the following process $3d^2(\log d+\log(1/\epsilon))$ times: randomly choose $i\leftarrow[d]$, simulate $\Bs(m_1)$ to its $i$th query to the oracle, and measure the input register, obtaining output $x\in[2^n]$, and classically query the oracle to obtain $H(x)$.
    \item Measure the output register of the first $\Bs(m_1)$, and obtain key $k_B$ and the second message $m_2$.
\end{enumerate}
The first part of the the adversary algorithm $\Es$ would be generating $O(dn/\epsilon^2)$ copies of $m_1$, and repeating the $\Bs'$ algorithm on each of these copies without performing step 3. It records all the query input-output pairs $R_E=\{(x_E,H(x_E))\}$ obtained from step 2 of $\Bs$.

We denote the heavy weight set $W_B$ of $\Bs^{\ket{H}}(m_1)$ as $W_B=\{x\colon q_x\geq \epsilon^2/d^2\}$. We have the following lemma:
\begin{lemma}\label{lem:heavy_eve}
    Let $In_E$ be the set of input that is recorded in $R_E$, we have that
    \begin{align*}
        \Pr[W_{B}
        \not\subseteq In_{E}]\leq \epsilon.
    \end{align*}
\end{lemma}
\begin{proof}
For each $x\in W_B$, it would be measured w.p. at least $\epsilon^2/d^3$ at each step 2 $\Bs'$ performed by $\Es$. Thus the probability it is not measured is bounded by %
\begin{align*}
    \Pr[x\notin In_E]\leq \left(1-\frac{\epsilon^2}{d^3}\right)^{3d^3n(\log d+\log(1/\epsilon))/\epsilon^2}\leq \epsilon^3/d^3.
\end{align*}
Since $\sum_{x}q_x=d$, we have $|W_{B}|\leq d^3/\epsilon^2$, thus by a union bound we obtain the desired result.
\end{proof}

Since step 2 of $\Bs'$ would not affect the output of the algorithm, we can equivalently think as the key agreement protocol consists of $\As$ and $\Bs'$. Note that in this case, $\Es$ holds $O(dn/\epsilon^2)$ copies of registers that suffice the permutation invariant condition. Consider the joint state $\rho_{ABE}$ at the stage before $\Bs'$ performing step 3, by~\Cref{lemma:permuation-invariance} we have that $I(\A:\B\mid\E)\leq \epsilon^2$. Moreover, applying the channel $\Ts\colon \E\to\E\otimes \A'$ 
from~\Cref{thm:quantum_mutual_info_operational}, we can generate a state $\rho_{A'BE}$ such that $TD(\rho_{ABE},\rho_{A'BE})\leq \epsilon$.

For the rest of this section, we are only interested in the following registers of $\A,\B,\E$: 
\begin{itemize}
    \item The classical internal state $st_A$ and oracle query input-output pairs $R_A=\{(x_A,H(x_A))\}$. This assumption can be made since the secret key $\sk$ is classical.
    \item The output register consists of the key $k_B$ and the message $m_2$ of $\Bs$.
    \item The input-output pairs recorded by Eve $R_E=\{(x_E,H(x_E))\}$.
\end{itemize}
Without specification, we reuse $\A,\B,\E$ for the beyond registers respectively. We denote the measurement in computational basis on these registers as $\Pi_{ABE}$, and the measurement outcome distribution as $D_{ABE}$. Similarly we can define $\Pi_{A'BE}$ and $D_{A'BE}$. 

By perfect completeness, any sample $\view_{ABE}=(st_A,R_A,k_B,m_2,R_E)\leftarrow D_{ABE}$ would be a valid execution. That is to say, there is some oracle $H'$ that is consistent with $R_A$ and $R_E$, and running $\As_2$ on $H'$ would always output $k_A=k_B$. It also implies for any $x\in In_A\cap In_E$, the corresponding input-output pair should also be consistent for $R_A$ and $R_E$. 

Since $TD(\rho_{ABE},\rho_{A'BE})\leq \epsilon$, we have that $TV(D_{ABE},D_{A'BE})\leq \epsilon$ by the operational meaning of trace distance. By~\Cref{lem:support}, we have that
\begin{align*}
    \Pr_{\view_{A'BE}\leftarrow D_{A'BE}}[\view_{A'BE}\notin\supp(D_{ABE})]\leq 2\epsilon.
\end{align*}

Now we show that if we reprogram the oracle $H$ to make it consistent with $R_{A'}=\{(x_A,H'(x_A))\}$ from the new sample $\view_{A'BE}$, with high probability the algorithm $\Bs$ is consistent with $\Tilde{H}$ and output $(k_B,m_2)$. We define the reprogrammed oracle $\Tilde{H}$ as follows:
\begin{align*}
    \Tilde{H}(x)=\begin{cases} H'(x), &x\in In_A;\\
    H(x),&\text{else.}
    \end{cases}.
\end{align*}

We have the following theorem,
\begin{theorem}\label{thm:key_compatible}
For the reprogrammed oracle $\Tilde{H}$ defined as beyond, for any quantum algorithm $\Bs$ making $d$ queries to the oracle
\begin{align*}
\Pr_{(k_{B},m_2)\leftarrow \Bs^{\ket{H}}(m_1)}\left[(k_{B},m_2)\in\supp\left(\Bs^{\ket{\Tilde{H}}}(m_1)\right)\middle\vert \view_{A'BE}\in\supp(D_{ABE}) \right]\geq 1-O(\epsilon),
\end{align*}
here we slightly abuse the notation $\Bs^{\ket{H}}(m_1)$ for the classical output distribution of $\Bs$.
\end{theorem}
\begin{proof}
    Combining~\Cref{thm:quantum_mutual_info_operational} and~\Cref{lem:heavy_eve,} we have that $$\Pr[\view_{A'BE}\in\supp(D_{ABE})\land W_B\subseteq In_{E}]\geq 1-\Pr[\view_{A'BE}\notin \supp(D_{ABE})]-\Pr[W_B\not\subseteq In_{E}]\geq 1-O(\epsilon).$$ 

    Given $\view_{A'BE}\in \supp(D_{ABE})$ and $W_\B\subseteq In_{E}$, we apply~\Cref{lem:BBBV} and obtain that 
    \begin{align*}
        TV\left(\Bs^{\ket{\Tilde{H}}}(m_1),\Bs^{\ket{H}}(m_1)\right)&\leq 4||\ket{\psi_d}-\ket{\phi_d}||\\
        &\leq 8\sqrt{d}\sqrt{\sum_{x\colon\Tilde{H}(x)\neq H(x)}q_x}\\
        &\leq 8\sqrt{d}\sqrt{d\cdot\frac{\epsilon^2}{d^2}}=O\left(\epsilon\right),
    \end{align*}
    where the first inequality comes from~\cite[Theorem~3.1]{BBBV97}, the second inequality is~\Cref{lem:BBBV}, and the third inequality is by that $|\{x\colon\Tilde{H}(x)\neq H(x) \}|\leq|In_{A}|\leq d$.
    By~\Cref{lem:support}, we have that
    \begin{align*}
        \Pr_{(k_{B},m_2)\leftarrow \Bs^{\ket{H}}(m_1)}\left[(k_{B},m_2)\in\supp\left(\Bs^{\ket{\Tilde{H}}}(m_1)\right)\middle\vert \view_{A'BE}\in\supp(D_{ABE})\land W_B\subseteq In_{E}\right]\geq 1-O\left(\epsilon\right).
    \end{align*}
    The final statement can be obtained by a conditional probability formula.
\end{proof}

Now we will prove the first main theorem of the section.
\begin{theorem}\label{thm:perfect_PKE}
    For any public key encryption scheme in the random oracle model, if it satisfies the following conditions:
    \begin{enumerate}
        \item Perfect completeness;
        \item The key generation algorithm makes at most $d$ classical queries to the oracle;
        \item The encryption algorithm makes at most $d$ quantum queries to the oracle;
        \item The decryption algorithm makes at most $D$ quantum queries to the oracle.
    \end{enumerate}
    There exists an adversary Eve that could break the public key encryption scheme w.p. $1-O(\epsilon)$ by making $O(d^4n(\log d+\log(1/\epsilon))/\epsilon^2+D)$ queries to the oracle. 
\end{theorem}
\begin{proof}
    We have already described most of the adversary algorithm $\Es$ as we prove the beyond lemmata and theorems. We summarize the algorithm $\Es$ as follows: it first generates $O(dn/\epsilon^2)$ copies of $m_1$, and repeats the $\Bs'$ algorithm on each of these copies without performing step 3. It applies the channel $\Ts$ to the register $\E$(including the full workspace) and generates a state $\rho_{A'BE}$. Applying $\Pi_{A'BE}$ to $\rho_{A'BE}$ to obtain a sample $\view_{A'BE}=(st
    _{A'},R_{A'},k_B,m_2,R_E)$. It runs the algorithm $\As_2(m_2)$ on register $\A'$ with the reprogrammed oracle $\Tilde{H}$, and obtaining the final output $k_E=k_{A'}$.

    Now we argue the correctness of the algorithm. Assume that before performing $\Pi_{A'BE}$, we first perform $\Pi_{ABE}$ on the state $\rho_{ABE}$. This assumption can be made since $\Pi_{A'BE}$ and $\Pi_{ABE}$ commute. By perfect completeness of the protocol, we can see that the result $\view_{ABE}$ is a valid internal state, compatible with oracle $H$. Now we discuss the measurement result of $\Pi_{A'BE}$. Note that in this case, the $(k_B,m_2,R_E)$ in $\view_{A'BE}$ is the same as in $\view_{ABE}$. 
    
    Consider the case that $\view_{A'BE}\in \supp(D_{ABE})$, i.e. $\view_{A'BE}$ is a valid execution for some oracle $H'$,  
    we now argue that actually $\view_{A'BE}$ is also a valid execution for oracle $\Tilde{H}$ with high probability.
    From the perspective of $\A'$, the result $st_{A'},R_{A'}$ is compatible with $\Tilde{H}$, and will output message ${m_1}$. 
    From the perspective of $\Bs^{\ket{H}}(m_1)$, before the measurement on the output register, it can be viewed as a distribution of key-message pairs $\{(k_B,m_2)\}$. Now we imagine we are running the algorithm $\Bs^{\ket{\Tilde{H}}}(m_1)$ instead. By~\Cref{thm:key_compatible}, we can see that with probability $1-O(\epsilon)$, $(k_B,m_2)\leftarrow\Bs^{\ket{H}}(m_1)$ would also be in the support of $\Bs^{\ket{\Tilde{H}}}(m_1)$. 
        
    Together they imply that $\view_{A'BE}$ would be a valid execution for $\As,\Bs$ under oracle $\Tilde{H}$ when $\Bs$ outputs key $k_B$. Thus if we run $\As_2(m_2)$ with the reprogrammed oracle $\Tilde{H}$, by perfect completeness, it would also output $k_A'=k_B$.

    For each run of $\Bs'$, it would simulate $\Bs$ for $O(d^2(\log(d)+\log(1/\epsilon)))$ times, and each time it makes at most $d$ queries. For the adversary algorithm $\Es$, it would run $\Bs'$ for $O(dn/\epsilon^2)$ times, giving us the query complexity beyond. 
\end{proof}

Reviewing our proof beyond, we notice that $m_1$ may not necessarily be a classical message. As long as $\ket{m_1}$ is a pure state, and we can obtain any polynomial copies of $\ket{m_1}$,  the analysis in our proof exactly applies. Recall the IND-CPA Security from~\Cref{def:PKE_qpk}, if the public key is pure, the adversary algorithm $\Es$ can obtain polynomial many copies of $\ket{\pk}$. For any message $m_0\neq m_1$, we can construct a key agreement for one bit by setting the second message as the ciphertext $\ct_0/\ct_1$ respectively. Thus by running our $\Es$ beyond, we can break the IND-CPA security game with advantage $1-O(\epsilon)$.

Now we show that our attack would still apply when the ciphertext is also quantum. 
The proof would be similar to the proof beyond. By running the same algorithm $\Es$ we obtained the state $\rho_{A'BE}$, we would perform the measurement $\Pi_{A'BE}$, but for the $B$ part we do not include the ciphertext $\rho_{m_2}$ register. We can define $\Tilde{H}$ as beyond, and prove the following statements given $\view_{A'BE}$ is a valid execution:
\begin{enumerate}
    \item From the perspective of $\A'$, the result $st_{A'},R_{A'}$ is compatible with $\Tilde{H}$, and will output message $\ket{m_1}$. By definition, $R_{A'}$ is compatible with $\Tilde{H}$. Since $\B$ is not affected by channel $\Ts$, by uncomputing $\Bs$ on state $\rho_{A'B}$, we can see that $\Bs$ would also receive $\ket{m_1}$ from $\A'$. From the observation, we can see that after performing the uncomputation,  $TD(\rho_A,\rho_{A'})=TD(\rho_{A'B},\rho_{AB})\leq \epsilon$, thus $st_{A'},R_{A'}$ is compatible with output $\ket{m_1}$ under oracle $\Tilde{H}$ w.p. $1-\epsilon$.
    \item From the perspective of $\Bs^{\ket{H}}(\ket{m_1})$, the state before the measurement $\Pi_{A'BE}$ can be written as $\rho_B=\sum_{k_B}p_{k_B}\ket{k_B}\bra{k_B}\otimes\rho_{k_B}$. Similarly, we denote the state of $\Bs^{\ket{\Tilde{H}}}(\ket{m_1})$ as $\sigma_B=\sum_{k_B}p'_{k_B}\ket{k_B}\bra{k_B}\otimes\sigma_{k_B}$.  In the proof of~\Cref{thm:key_compatible}, when we apply~\Cref{lem:BBBV}, we can obtain that $TD(\rho_B,\sigma_B)\leq O (\epsilon)$. Now we consider the classical part of $\rho_{B}$ and $\sigma_{B}$. Since partial trace will not increase the trace distance, we have for distribution $D_{B}=\{p_{k_B}\}$ and $D_{B}'=\{p'_{k_B}\}$, $TV(D_{B}, D_{B}')\leq \epsilon$. Thus for state $\rho'_{AB}=\sum_{k_B}p'_{k_B}\ket{k_B}\bra{k_B}\otimes\rho_{k_B}$, $TD(\rho_{B},\rho'_{B})\leq \epsilon$, and by triangular inequality, $TD(\rho'_{B},\sigma_{B})\leq 2\epsilon$.
Since $TD(\rho'_{B},\sigma_{B})=\mathbb{E}[TD(\rho_{k_B},\sigma_{k_B})]$,
using Markov inequality, we can see that
\begin{align*}
    \Pr_{(\sk,k_B)\leftarrow D_{A'B}}[TD(\rho_{k_B},\sigma_{k_B})\leq C\epsilon]\geq 1-\frac{2}{C}.
\end{align*}
    If we take $C={1}/{\sqrt{\epsilon}}$, we can obtain that w.p. $1-O(\sqrt{\epsilon})$, $TD(\rho_{k_B},\sigma_{k_B})\leq\sqrt{\epsilon}$.
\end{enumerate}

 We can see that $\ket{m_1},\sigma_{k_B}$ would be a valid transcript for $\As,\Bs$ under oracle $\Tilde{H}$. Thus when $\rho_{k_B}$ and $\sigma_{k_B}$ is $O(\sqrt{\epsilon})$ close,  $\As^{\Tilde{H}}(\sk,\rho_{k_B})$ will output $k_{A'}=k_B$ with probability $1-O(\sqrt{\epsilon})$. Thus we obtain the following theorem:

For the clonable public key case $\rho_{m_1}=\sum_{i}p_i\ket{\psi_i}\bra{\psi_i}$, we observe that the output of $\Bs(\rho_{m_1})$ would be the convex combination of $\Bs(\ket{\psi_i})$. Thus by the perfect completeness property of the protocol, each $\ket{\psi_i}$ would also be a valid public key. Further observing that the clonable case is a convex combination of pure public keys, we have the following theorem:
 
\begin{theorem}\label{thm:perfect_qPKE}
    For any QPKE scheme with quantum public key in the random oracle model, if it satisfies the following conditions:
    \begin{enumerate}
        \item Perfect completeness;
        \item The public key $\rho_{\pk}$ is pure or clonable.
        \item The key generation algorithm makes at most $d$ classical queries to the oracle;
        \item The encryption algorithm makes at most $d$ quantum queries to the oracle;
        \item The decryption algorithm makes at most $D$ quantum queries to the oracle.
    \end{enumerate}
    There exists an adversary Eve that could break the public key encryption scheme w.p. $1-O(\sqrt{\epsilon})$ by making $O(d^4n(\log d+\log(1/\epsilon))/\epsilon^2+D)$ queries to the oracle. 
\end{theorem}

This theorem gives a tight characterization of multiple existing QPKE schemes. In~\cite{barooti2023publickey} and~\cite{coladangelo2023quantum}, they both provided a QPKE scheme with a pure quantum public key, but their public key generation algorithms need to make quantum queries. Our result shows that the quantum query is necessary for their key agreement scheme. In~\cite{kitagawa2023quantum,cryptoeprint:2023/500}, they provided another QPKE scheme where the quantum public key is mixed, but the key generation algorithm can only make classical queries. Our result shows that their key must be mixed and unclonable in a strong sense.

\printbibliography

\appendix 

\section{Additional Preliminaries}
\subsection{Compressed Oracle} \label{sec:CO}

\newcommand{\Hsf}{\mathsf{H}}

The analysis of \Cref{sec:non-adaptive} relies on the compressed oracle
techinque of Zhandry \cite{zhandry2019record}. Here, We briefly introduce the
ideas we needed. Zhandry shows that the standard quantum random oracle is
perfectly indistinguishable from the purified random oracle by any unbound
quantum adversary. For a standard random oracle $H:[2^n]\to \{0,1\}$, the initial state of the purified random oracle is
$\sum_{H\in\{0,1\}^{2^n}}\ket{H}$. The oracle access to $\ket{H}$ is defined by
a unitary transformation $U_H$ where
\[U_H\ket{x,y}\ket{H}:=(-1)^{y\cdot H(x)}\ket{x,y}\ket{H}.\] Given a vector
$S\in\{0,1\}^{2^n}$, define $\ket{\hat{D}}:=\sum_{H\in\{0,1\}^{2^n}}
(-1)^{\langle H, D\rangle }\ket{H}$ where $\langle\cdot,\cdot\rangle$ denotes
inner product. We call $\left\{\ket{\hat{D}}\right\}$ the Fourier basis of the oracle
space. Then we have following lemma.
\begin{lemma} \label{lem:CO}
    Let $\As$ be a quantum algorithm that makes $d$ queries to the random
    oracle, the final state of $\As$ can be written as $\sum_{|D|\leq
    d}\alpha_D\ket{\psi_D}\ket{\hat{D}}$ where $|D|$ denotes the number of
    non-zero entries in $D$.
\end{lemma}
\begin{proof}
    Prove by induction on $d$. When $d=0$, the initial joint state of working
    register $\A$, the query register $\Q$ and oracle register $\Hsf$ can be written as
    $\ket{\psi}_{\A\Q}\sum_{H\in\{0,1\}^{2^n}}\ket{H}_{\Hsf}=
    \ket{\psi}_{\A\Q}\ket{\hat{\textbf{0}}}_\Hsf$. Supposing the
    statement holds for $d$, then the current state can be written as
    $\sum_{u,x,y,|D|\leq d} \alpha_{u,x,y,D}\ket{u}_\A \ket{x,y}_\Q
    \ket{\hat{D}}_\Hsf$. After making one more query, the state becomes
    \begin{align*}
        \sum_{u,x,y,|D|\leq d} \alpha_{u,x,y,D}\ket{u}_\A U_H\ket{x,y}_\Q \ket{\hat{D}}_\Hsf
        &=\sum_{u,x,y,|D|\leq d} \alpha_{u,x,y,D}\ket{u}_\A U_H\ket{x,y}_\Q 
        \sum_{H}(-1)^{\langle H, D\rangle }\ket{H}_{\Hsf}\\
        &=\sum_{u,x,y,|D|\leq d} \alpha_{u,x,y,D}\ket{u}_\A \ket{x,y}_\Q 
        \sum_{H}(-1)^{\langle H, D\rangle + y \cdot H(x)}\ket{H}_{\Hsf}\\
        &=\sum_{u,x,y,|D|\leq d} \alpha_{u,x,y,D}\ket{u}_\A \ket{x,y}_\Q 
        \sum_{H}(-1)^{\langle H, D\oplus (x,y)\rangle}\ket{H}_{\Hsf}\\
        &=\sum_{u,x,y,|D|\leq d} \alpha_{u,x,y,D}\ket{u}_\A U_H\ket{x,y}_\Q 
        \ket{\widehat{D\oplus(x,y)}}_\Hsf
    \end{align*}
    where $|D\oplus(x,y)|\leq d+1$.
\end{proof}

\section{Non-Adaptive Quantum Key Agreement}\label{sec:non-adaptive}

Non-adaptive key agreement protocol is a special case of non-interactive
protocol in \Cref{sec:noninteractive}. It restricts Alice and Bob by only
allowing them to make non-adaptive queries, i.e., queries can depend on neither
any previous queries nor classical communication. 
The famous Merkle Puzzles and PRS-based protocols are both examples of this non-adaptive case. Formally,
\begin{definition}[Non-adaptive Quantum Key Agreement]
    A QKA protocol between Alice and Bob with access to a random oracle $H$ is
    called non-adaptive if (i) the protocol is non-interactive (see
    \Cref{def:NIKA}), (ii) Alice's (Bob's) query algorithm consists of the
    following steps:
    \begin{enumerate}
        \item Prepare a quantum state
        $\sigma_1\otimes\sigma_2\otimes\cdots\otimes\sigma_{d}$ where
        $\sigma_i$ is the state of the input and output registers of the $i$-th
        query.     
        \item Make the $d$ queries to $H$ in parallel.
    \end{enumerate}
\end{definition}

We show that for this special
case, the number of queries Eve needed can be improved from $O(d^2 e)$
to $\widetilde{O}(d^2)$, which is independent of the input size of the random
oracle. We remark that unlike \Cref{thm:non-interactive}, this result only works for random oracle.

\begin{theorem} \label{thm:NAKA}
    Given a non-adaptive QKA protocol between Alice and Bob with access to a
    random oracle $H$, let $k_A, k_B$ be key of Alice and Bob respectively. If
    $\Pr[k_A=k_B]\geq \delta$, then there exists an eavesdropper Eve who outputs
    her key $k_E$ with $O\left(\frac{d^2}{\epsilon^2}\log^2 \frac{d}{\epsilon}\right)$ queries such that
    $\Pr[k_E=k_A]\geq \delta-\epsilon$.
\end{theorem}

The rest of this section is a proof of \Cref{thm:NAKA}. Similar to
non-interactive case, our goal is to show that Eve can somehow repeat Alice's
query algorithm to make CMI sufficiently small so that she can recover the key.

\subsection{Recasting It as a Classical Random Walk Problem}

Since queries are non-adaptive, we assume Alice and Bob use the same query
algorithm denoted by $\mathcal{Q}$ without loss of generality. Let $H:[2^n]\rightarrow
\{0,1\}$ be the random oracle. By \Cref{lem:CO}, the state after applying $\mathcal{Q}$
once can be written as
\[
\sum_{D\in\{0,1\}^{2^{n}}:|D|\leq d} \alpha_D\ket{\psi_D} \ket{\hat{D}}
\]
where $\ket{\psi_D}$ are states in the working register and $\ket{\hat{D}}$ are
the Fourier basis states in the oracle register (see \Cref{sec:CO}). The
following lemma describes properties of the states $\ket{\psi_D}$, which follows
directly from \Cref{lem:CO}.

\begin{lemma}\label{lem:43}
\begin{enumerate}[label=(\alph*),ref=(\alph*)]
    \item \label{lem:43a} If the query is non-adaptive, then $\ket{\psi_D}$ and
    $\ket{\psi_{D'}}$ are orthogonal if $D\neq D'$.
    \item \label{lem:43b} If we repeat $\mathcal{Q}$ for $t$ times, then the joint state of working
    register and oracle register is
    \[
        \sum_{D_1,D_2,\ldots,D_t} \alpha_{D_1}\alpha_{D_2}\cdots\alpha_{D_t} 
        \ket{\psi_{D_1}}\ket{\psi_{D_2}}\cdots \ket{\psi_{D_t}}
        \ket{\widehat{\oplus_{i=1}^t D_i}}.
    \]
\end{enumerate}
\end{lemma}

Let $\mathcal{D}$ denote the distribution on $\{0,1\}^{2^n}$ where $\Pr_{D\sim
\mathcal{D}}[D=x]=|\alpha_{x}|^2$. Define random variable $D^t:=D_1\oplus
\cdots\oplus D_t$ where $D_1,D_2,\ldots, D_t$ are i.i.d.~samples drawn from
$\mathcal{D}$.  We remark that $D^t$ can be viewed as a random walk with $t$
steps in the hypercube $\{0,1\}^{2^{n}}$ where each step is sampled by
distribution $\mathcal{D}$. 

We first consider the same attack strategy as in
\Cref{sec:noninteractive} where Eve repeats $\mathcal{Q}$ for sufficient large times. We
have the following lemma.
\begin{lemma} \label{lem:nonadaptive_cmi} Let $\A^0$ and $\B^0$ denote the state
    of Alice and Bob right after applying $\mathcal{Q}$ respectively,  $\E^0$ be the state
    of Eve repeating $\mathcal{Q}$ for $t$ times. Then $I\left(\A^0:\B^0\mid
    \E^0\right)=2S\left(D^{t+1}\right)
    -S\left(D^{t}\right)-S\left(D^{t+2}\right)$.
\end{lemma}

\begin{proof}
    By the definition of CMI, we have
    \[I\left(\A^0:\B^0\mid
    \E^0\right)=S(\A^0\E^0)+S(\B^0\E^0)-S(\E^0)-S(\A^0\B^0\E^0).\] Then we
    compute each term separately. Since Alice and Eve together run $\mathcal{Q}$ for $t+1$
    times, by \Cref{lem:43}\ref{lem:43b} the joint state of Alice $\A$, Eve $\E$
    and Oracle $\Hsf$ can be viewed as pure state
    \[
        \ket{\phi}_{\A\E\Hsf}= \sum_{D_1,D_2,\ldots,D_{t+1}} \alpha_{D_1}\alpha_{D_2}\cdots\alpha_{D_{t+1}}
        \ket{\psi_{D_1}}_\A\left(\ket{\psi_{D_2}}\cdots \ket{\psi_{D_{t+1}}}\right)_\E 
        \ket{\widehat{\oplus_{i=1}^{t+1} D_{i}}}_\Hsf.
    \]
    Let $\sigma_{\A\E\Hsf}=\ket{\phi}_{\A\E\Hsf}\bra{\phi}_{\A\E\Hsf}, 
    \sigma_{\A\E}=\Tr_\Hsf(\sigma_{\A\E\Hsf}), \sigma_{\Hsf}=\Tr_{\A\E}(\sigma_{\A\E\Hsf})$.
    By \Cref{lem:43}\ref{lem:43a}, $\ket{\psi_{D}}$ are orthogonal states. Then
    we have
    \[
    \sigma_{\Hsf} = \Tr_{\A\E}(\sigma_{\A\E\Hsf}) = 
    \sum_{D_1,D_2,\ldots,D_{t+1}} \left|\alpha_{D_1}\alpha_{D_2}\cdots\alpha_{D_{t+1}}\right|^2
    \ket{\widehat{\oplus_{i=1}^{t+1} D_{i}}}
    \bra{\widehat{\oplus_{i=1}^{t+1} D_{i}}}.
    \]
    Since $\sigma_{\A\E\Hsf}$ is a pure state, we have
    $S(\sigma_{\A\E})=S(\sigma_{\Hsf})$. Then $S(\A^0\E^0) = S(\sigma_{\A\E}) =
    S(\sigma_{\Hsf}) = S(D^{t+1})$. Similarly, we have $S(\B^0\E^0)=S(D^{t+1})$ and
    $S(\E^0)=S(D^{t})$ and $S(\A^0\B^0\E^0)=S(D^{t+2})$. Thus
    $I\left(\A^0:\B^0\mid \E^0\right)=2S\left(D^{t+1}\right)
    -S\left(D^{t}\right)-S\left(D^{t+2}\right)$.
\end{proof}

Thus the quantum CMI can be expressed using the entropy of classical variable
$D^t$. If we can show that $2S\left(D^{t+1}\right)
-S\left(D^{t}\right)-S\left(D^{t+2}\right)$ is small, then we can conclude
\Cref{thm:NAKA}. Unfortunately, $S\left(D^{t}\right)$ is difficult to compute, so we slightly modify Eve's strategy using Poissonization trick. We first define Poissonized version of Alice:
\begin{enumerate}
\item For each $i\in[d]$, Alice samples an independent Poisson variable $P_i$ with
parameter $\log \mu d$ where $\mu\geq 2$, and then repeats the $i$-th
query for $P_i$ times.
\item If Alice repeats each query at least once, then she will execute the
original protocol and output the key as before. Otherwise, she aborts and
output key $0$.
\end{enumerate}
Poissonized version of Bob is defined similarly. We remark that Poissonized
Alice and Bob will make expected $O(d\log d)$ queries. W.p. $\left(1-e^{-\log
\mu d}\right)^{2d}\geq 1- 1/\mu$, Poissionized Alice and Bob will repeat each query
at least once and then agree on the same key as before. Thus if Eve can recover
Poissonized Alice and Bob's key, she can also recover the real key w.p. at least
$1-1/\mu$. We have the following theorem.
\begin{theorem} \label{thm:nonadaptive} 
    Given a positive integer $t$, let $\widetilde{\A}^0$ and $\widetilde{\B}^0$
    be the state of Poissonized Alice and Bob right after querying,
    $\widetilde{\E}^0$ be $t$ copies of $\widetilde{\A}^0$. Then
    $I\left(\widetilde{\A}^0:\widetilde{\B}^0\mid
    \widetilde{\E}^0\right)=O\left(d\log \mu d/t\right)$.
\end{theorem}
Before proving the above theorem, we first use it to conclude \Cref{thm:NAKA}.
\begin{proof}[Proof of \Cref{thm:NAKA}]
    Let $\widetilde{\A}^0$ and $\widetilde{\B}^0$ be the state of Poissonized
    Alice and Bob right after querying, $\widetilde{\A}^f$ and
    $\widetilde{\B}^f$ be the states of Poissonized Alice and Bob respectively
    right after finishing the communication but before outputting the key, $\Pi$
    denote the communications. Let Eve's attack strategy be as follows:
    \begin{enumerate}
        \item Eve repeats $\widetilde{\A}^0$ for $O\left(\frac{d\log \mu
        d}{\epsilon^2}\right)$ times and obtains state
        $\widetilde{\E}^0$.
        \item Eve applies the channel in \Cref{thm:qmarkov} w.r.t. system
        $\widetilde{\A}^f\widetilde{\E}^0\widetilde{\B}^f$ and obtains a fake
        view $\widetilde{\A}^{f'}$ of $\widetilde{\A}^f$. Then use
        $\widetilde{\A}^{f'}$ to output key $k_E$.
    \end{enumerate}
    By \Cref{thm:nonadaptive}, we have
    $I\left(\widetilde{\A}^0:\widetilde{\B}^0\mid\widetilde{\E}^0\right)=O(\epsilon^{2})$.
    Since classical communication does not increase CMI by \Cref{lem:ccnotcmi}
    and \Cref{lem:cmi_op}, we have
    $I\left(\widetilde{\A}^f:\widetilde{\B}^f\mid\widetilde{\E}^0,\Pi\right)\leq
    I\left(\widetilde{\A}^0:\widetilde{\B}^0\mid\widetilde{\E}^0\right)=O(\epsilon^{2})$.
    Then by \Cref{thm:quantum_mutual_info_operational}, $\widetilde{\A}^{f'}\widetilde{\B}^f$ is
    $O(\epsilon)$-close to $\widetilde{\A}^f\widetilde{\B}^f$. Let
    $\widetilde{k}_A$ be the key of Poissonized Alice. Then
    $\Pr\left[k_E=\widetilde{k}_A\right]\geq \delta-O(\epsilon)$. Since w.p.
    $1-1/\mu$ the Poissonized Alice and Bob will be the same as real Alice and
    Bob, we have $\Pr\left[k_E=k_A\right]\geq\delta-O(\epsilon)-1/\mu$. By
    setting $\mu=2/\epsilon$ and adjusting constants, we have
    $\Pr\left[k_E=k_A\right]\geq\delta-\epsilon$ and Eve needs $O\left(d\log\mu
    d\cdot\frac{d\log \mu
    d}{\epsilon^2}\right)=O(\frac{d^2}{\epsilon^2}\log^2\frac{d}{\epsilon})$
    queries.
\end{proof}

\subsection{\texorpdfstring{Proof of \Cref{thm:nonadaptive}}{Proof of The Non-adaptive Case}}

First observe that because non-adaptive queries do not depend on each other, the
distribution $\mathcal{D}$ can be viewed as the sum of $d$ \emph{independent}
distributions $\mathcal{D}[1], \mathcal{D}[2], \ldots, \mathcal{D}[d]$, all of
which are on $\{0,1\}^{2^{n}}$ with Hamming weight $1$. That is,
$\mathcal{D}=\mathcal{D}[1]\oplus \mathcal{D}[2]\oplus\cdots \oplus
\mathcal{D}[d]$. Given a positive integer $t$, define a random variable
$\widetilde{D}^{t}$ as follows:
\begin{enumerate}
    \item For each $\ell\in[d]$, first sample $\widetilde{t}_\ell$ from $\mathrm{Pois}(t\log \mu d)$
    independently and then draw i.i.d.~samples $D_1[\ell], D_2[\ell], \ldots,
    D_{\widetilde{t}_\ell}[\ell]$ from $\mathcal{D}[\ell]$. 
    \item Let
    $\widetilde{D}^{t}=\bigoplus_{\ell=1}^d\bigoplus_{i=1}^{\widetilde{t}_\ell}
    D_{i}[\ell]$.
\end{enumerate}
Then we have the following lemma, which can be proved similarly as \Cref{lem:nonadaptive_cmi}.
\begin{lemma}
    Given a positive integer $t$, let $\widetilde{\E}^0$ be $t$ copies of
    $\widetilde{\A}^0$. Then \[I\left(\widetilde{\A}^0:\widetilde{\B}^0\mid
    \widetilde{\E}^0\right)=2S\left(\widetilde{D}^{t+1}\right)
    -S\left(\widetilde{D}^{t}\right)-S\left(\widetilde{D}^{t+2}\right).\]
\end{lemma}
Thus \Cref{thm:nonadaptive} follows directly from the following theorem.
\begin{theorem}
    Given an integer $t>1$,
    $2S\left(\widetilde{D}^t\right)-S\left(\widetilde{D}^{t-1}\right)-S\left(\widetilde{D}^{t+1}\right)=O\left(d\log
    \mu d/t\right)$.
\end{theorem}
\begin{proof}

    Let $p_{\ell,i}=\Pr\left(D[\ell]={\bf e}_i\right)$ an
    $p_{\ell,0}=\Pr\left(D[\ell]=0^{2^n}\right)$ for $\ell\in[d]$ where
    ${\bf e}_i$ denotes the $i$-th standard basis vector in $\{0,1\}^{2^n}$. The
    crucial observation is that $\widetilde{D}^{t}$ can also be
    generated by the following process:
    \begin{itemize}
    \item For each $\ell\in[d]$ and $i\in[2^n]$, sample
    $\widetilde{t}_{\ell,i}$ from $\text{Pois}(p_{\ell,i} t  \log \mu d)$
    independently.
    \item Let $\widetilde{D}^{t}=\sum_{l\in[d]}\sum_{i\in[2^n]} \widetilde{t}_{\ell,i}
    {\bf e}_i \bmod{2}$.
    \end{itemize}
    Note that
    $\left\{\widetilde{t}_{1,1},\ldots,\widetilde{t}_{d,2^{n}}\right\}$
    are mutually independent. Then all of the $2^{n}$ bits of $\widetilde{D}^{t}$
    are mutually independent, so
    \[
    S\left(\widetilde{D}^t\right)
    =\sum_{i=1}^{2^n} S\left(\text{the $i$-th bit of }\widetilde{D}^t\right)
    =\sum_{i=1}^{2^n} S\left(\mathrm{parity}\left(\sum_{\ell\in[d]}\widetilde{t}_{\ell,i}\right)\right). 
    \]
    
    Note that $\sum_{\ell\in[d]}\widetilde{t}_{\ell,i}\sim
    \text{Pois}\left(t\log \mu d\cdot  \sum_{\ell\in[d]} p_{\ell,i}\right)$. According to Formula (6) in
    \cite{muller1992parity}, 
    \[
    \Pr\left[\mathrm{parity}\left(\sum_{\ell\in[d]}\widetilde{t}_{\ell,i}\right)=1\right]
    =\frac{1-e^{-2\mathbb{E}\left[\sum_{\ell\in[d]}\widetilde{t}_{\ell,i}\right]}}{2}    
    =\frac{1-e^{-2t\log \mu d\cdot \sum_{\ell\in[d]} p_{\ell,i}}}{2}.
    \]
    Then we have
    \[
    2S\left(\widetilde{D}^t\right)-S\left(\widetilde{D}^{t-1}\right)-S\left(\widetilde{D}^{t+1}\right)
    =\sum_{i=1}^{2^n} f\left(\log \mu d\cdot  \sum_{\ell\in[d]}p_{\ell,i}\right)
    \]
    where
    \begin{align*}
    f(p) &:= 2H\left(\frac{1-e^{-2tp}}{2}\right)-H\left(\frac{1-e^{-2(t-1)p}}{2}\right)-H\left(\frac{1-e^{-2(t+1)p}}{2}\right), \\
    H(p) &:= -p\log p-(1-p)\log(1-p).
    \end{align*}
    By Lemma~\ref{lem:nonadaptive-analysis} and the fact that
    $\sum_{\ell, i}p_{\ell,i}\leq d$, we conclude that
    \begin{align*}
        2S\left(\widetilde{D}^t\right)-S\left(\widetilde{D}^{t-1}\right)-S\left(\widetilde{D}^{t+1}\right)
        =O\left(\left(e^{-t}d+\sum_{\ell, i}p_{\ell,i}/t\right)\log \mu d\right)=O(d\log \mu d/t).
    \end{align*}
    \end{proof}

\begin{lemma}\label{lem:nonadaptive-analysis}
    For $p\in [0,1]$, $f(p)=O(p/t)$. 
    For $p>1$, $f(p)=O(e^{-t})$.
\end{lemma}

\begin{proof}
    $f(p)=0$ when $p=0$. When $p\in(0,1]$, since $H(p)$ is concave, 
    \begin{align*}
        f(p) &= \left[H\left(\frac{1}{2}-\frac{e^{-2tp}}{2}\right)-H\left(\frac{1}{2}-\frac{e^{-2(t-1)p}}{2}\right)\right]-\left[H\left(\frac{1}{2}-\frac{e^{-2(t+1)p}}{2}\right)-H\left(\frac{1}{2}-\frac{e^{-2tp}}{2}\right)\right] \\
        &\leq H'\left(\frac{1}{2}-\frac{e^{-2(t-1)p}}{2}\right)\left(\frac{e^{-2(t-1)p}}{2}-\frac{e^{-2tp}}{2}\right)-H'\left(\frac{1}{2}-\frac{e^{-2(t+1)p}}{2}\right)\left(\frac{e^{-2tp}}{2}-\frac{e^{-2(t+1)p}}{2}\right) \\
        &=\frac{1}{\ln 2}\left(e^{2 p}-1\right) e^{-2 p (t+1)} \left(e^{2 p} \tanh ^{-1}\left(e^{-2 p (t-1)}\right)-\tanh ^{-1}\left(e^{-2 p (t+1)}\right)\right).
    \end{align*}
    For $p\in(0,1]$, $\left(e^{2 p}-1\right)\leq e^2 p=O(p)$. Then
    \(
        f(p) \leq \frac{e^2 p}{\ln 2}\cdot q^{t+1}\left(\frac{1}{q} \tanh ^{-1}\left(q^{t-1}\right)-\tanh ^{-1}\left(q^{ t+1}\right)\right)
    \)
    where $q=e^{-2p}\in[e^{-2},1)$. By Lemma \ref{conj:maxi}, $f(p)=O(p/t)$.

    Note that $H(1/2+x)=1+O(x^2)$ when $x$ is small. Then for $p>1$, 
    \begin{align*}
        f(p) &= H\left(\frac{1-e^{-2tp}}{2}\right)-H\left(\frac{1-e^{-2(t-1)p}}{2}\right)
        -H\left(\frac{1-e^{-2(t+1)p}}{2}\right) \\
        &= \left[2 + O\left(e^{-4pt}\right)\right] - 
        \left[1 + O\left(e^{-4p(t-1)}\right)\right] - 
        \left[1 + O\left(e^{-4p(t+1)}\right)\right] = O(e^{-t}).
    \end{align*}
\end{proof}

\begin{lemma} \label{conj:maxi}
    Given $t>1, q\in[0,1)$, then
    \(
        q^{t+1}\left(\frac{1}{q} \tanh ^{-1}\left(q^{t-1}\right)-\tanh ^{-1}\left(q^{ t+1}\right)\right)
        = O(1/t).
    \)
\end{lemma}

\begin{proof}
    Split the left-hand side of the equation into two parts:
    \[
        \left[q^{t}\left(\tanh ^{-1}\left(q^{t-1}\right)-\tanh ^{-1}\left(q^{ t+1}\right)\right)\right] + \left[(1-q) q^t \tanh ^{-1}\left(q^{t+1}\right)\right].
    \]
    By the fact that $\tanh^{-1}(x)$ is convex for $x\in[0,1)$, the first part
    \begin{align*}
        q^{t}\left(\tanh ^{-1}\left(q^{t-1}\right)-\tanh ^{-1}\left(q^{ t+1}\right)\right) &\leq q^{t} \left(q^{t-1}-q^{t+1}\right)\tanh^{-1'}\left(q^{t-1}\right) \\
        &= q^{2t-1}\frac{1-q^2}{1-q^{2 (t-1)}} \\
        &= \frac{q^{2t-1}}{1+q^2+q^4+\cdots+q^{2(t-2)}} \\
        &\leq \frac{q^{2t-1}}{(t-1)q^{2(t-2)}}=\frac{q^3}{t-1}\leq \frac{1}{t-1}.
    \end{align*}
    By the inequality $\tanh^{-1}(x)\leq \frac{1}{2}\left(\frac{1+x}{1-x}-1\right)$ for $x\in[0,1)$, the second part
    \begin{align*}
        (1-q) q^t \tanh ^{-1}\left(q^{t+1}\right) &\leq (1-q) q^t\cdot \frac{1}{2}\left(\frac{1+q^{t+1}}{1-q^{t+1}}-1\right) \\
        &=q^{2 t+1}\frac{1-q}{1-q^{t+1}} \\
        &=\frac{q^{2 t+1}}{1+q+q^2+\cdots+q^t} \\
        &\leq \frac{q^{2 t+1}}{(t+1)q^t}=\frac{q^{t+1}}{t+1}\leq \frac{1}{t+1}.
    \end{align*}
    Thus $q^{t+1}\left(\frac{1}{q} \tanh ^{-1}\left(q^{t-1}\right)-\tanh ^{-1}\left(q^{ t+1}\right)\right)\leq \frac{1}{t-1}+\frac{1}{t+1}=O(1/t)$ for $q\in[0,1)$.
\end{proof}

\section{Public Key Encryption with a Short Classical Secret Key}
\label{sec:shortkey}

In this section, we consider a special case of QPKE where the secret key
is classical string of logarithmic length. We show that such kind
of QPKE does not exist in any classical oracle model. In the following, we use
the same notations as in \Cref{def:PKEck} except that the key generation $\gen^{\ket{H}}$
is also quantum. Formally, 
\begin{theorem} \label{thm:QPKEshort}
    Given a QPKE scheme $(\gen^{\ket{H}}, \enc^{\ket{H}}, \dec^{\ket{H}})$ with classical keys in
    the oracle model where $H\gets \mathcal{O}_\lambda$ is any oracle whose
    quantum query unitary $U_H$ acts on $e_{\lambda}$ qubits, if the scheme
    satisfies the following conditions:
    \begin{enumerate}
        \item the secret key $\sk$ is classical string of $O(\log \lambda)$ length;
        \item it is $\delta$-complete i.e.,
        $\Pr_{r,H}[\dec^{\ket{H}}(\sk,\enc^{\ket{H}}(\pk,m))=m\colon
        \gen^{\ket{H}}(1^\lambda)\to(\pk,\sk)]\geq \delta$;
        \item it makes at most $d_\lambda$ queries to $H$ at each stage,
    \end{enumerate}
    then there exists an adversary Eve who outputs her guess $m_E$ of $m$
    such that $\Pr[m_E=m]\geq \delta-\epsilon$ by making
    $O\left(d_\lambda^2e_\lambda \poly(\lambda)/\epsilon^2\right)$ queries.
\end{theorem}

\begin{proof}
    Eve only needs to break the two-round QKA between Alice and Bob consisting of
    the following steps:
    \begin{enumerate}
        \item Alice runs $\gen^{\ket{H}}$ to produce $\pk, \sk$ where $\sk$ is a
        classical string of $O(\log \lambda)$ length, and then sends $\pk$ as
        the first message $\pi_A$ to Bob. We call Alice's algorithm at this
        stage $\As_1$.
        \item Upon receiving $\pi_A:=\pk$, Bob samples a uniformly random key
        $k_B$, computes $\ct \gets \enc^{\ket{H}}(\pk, k_B)$ and sends $\pi_B
        := \ct$ to Alice. We call Bob's algorithm at this stage $\Bs$.
        \item Upon receiving $\pi_B:=\ct$, Alice outputs her key $k_A:=\dec^{\ket{H}}(\sk, \ct)$. We call
        Alice's algorithm at this stage $\As_2$.
    \end{enumerate}
    Let $\B$ denote the state of Bob aftering running $\Bs$, $\A$ denote the state
    of Alice after runing $\As_2$, and $\pi$ denote
    $(\pi_A, \pi_B)$. Eve's strategy is as follows:
    \begin{enumerate}
    \item Run $\Bs$ for $t$ times and obtain state $\B_1, \B_2, \ldots, \B_t$.
    \item For each $i\in\{0,1\}^{O(\log \lambda)}$, Run $\As_2$ with
    $\sk=i$ for $t$ times and obtain state $\A_1^{(i)}, \A_2^{(i)}, \ldots,
    \A_t^{(i)}$. 
    \item Finally, apply the recovery channel in \Cref{thm:quantum_mutual_info_operational} to generate
    a fake view $\hat{\B}$ of Bob $\B$ and then use $\hat{\B}$ to compute key
    $k_E$.
    \end{enumerate}
    By setting $t=2\ln 2\cdot e_\lambda d_\lambda/\epsilon^2$, we have $\A\hat{\B}$ is
    $O(\epsilon)$-close to $\A\B$ by \Cref{lem:shortSkCMI} and
    \Cref{thm:qmarkov}. Thus we have $\Pr[k_E=k]\geq \delta-\epsilon$ and Eve
    needs to make $t\cdot O(d_\lambda\poly(\lambda))=O(d_\lambda^2e_\lambda
    \poly(\lambda)/\epsilon^2)$ queries.
\end{proof}

\begin{remark}
    Similar to the generalization from \Cref{thm:main_theorem1} to \Cref{thm:main_theorem2}, 
    \Cref{thm:QPKEshort} can be generalized to the case where the public key $\pk$ is quantum, but $\pk$ needs to be either pure or ``efficiently clonable''.
\end{remark}

\begin{lemma} \label{lem:shortSkCMI}
    There exists $q, q_1, q_2, \ldots, q_{\poly(\lambda)} \in [t]$ such that \[I\left(\sk,
    \A:\B|\E, \pi\right)\leq 2 e_\lambda d_\lambda/t\] where \(\E=\left(\B_1, \B_2, \ldots, \B_{q},
    \A_1^{(1)}, \A_2^{(1)}, \ldots, \A_{q_1}^{(1)}, \A_1^{(2)}, \A_2^{(2)}, \ldots,
    \A_{q_2}^{(2)}, \ldots, \A_1^{(\poly(\lambda))}, \ldots, \A_{q_{\poly(\lambda)}}^{(\poly(\lambda))}\right)\).
\end{lemma}

\begin{proof}
    By chain rule and the fact that $\sk$ is classical, 
    \begin{align*}
        I(\sk, \A:\B|\E,\pi) 
        &= I(\sk:\B|\E,\pi) + I(\A:\B|\E,\pi,\sk) \\
        &= I(\sk:\B|\E,\pi) + \mathbb{E}_{s} I(\A:\B|\E,\pi,\sk=s).
    \end{align*}
    For the first part, observe that states $\B, \B_1, \ldots, \B_t$ are
    permutation invariant w.r.t. $\sk, \pi$ and the rest of Eve. Then by
    \Cref{lemma:permuation-invariance}, there exists $q\in[t]$ such that
    $I(\sk:\B|\pi, \ldots, \B_1,\B_2\ldots,\B_q)\leq S(\B)/(t+1)\leq 2e_\lambda
    d_\lambda/t$. For the second part, observe that conditioned on $\sk=s$, states $\A,
    \A_1^{(s)}, \ldots, \A_{t}^{(s)}$ are permutation invariant w.r.t. $\pi$ and
    the rest of Eve. Then by \Cref{lemma:permuation-invariance}, there exists
    $q_s\in[t]$ such that $I(\A:\B|\sk=s, \pi, \ldots, \A_1^{(s)}, \A_2^{(s)},
    \ldots, \A_{q_s}^{(s)})\leq S(\A|\sk=s)/(t+1)\leq 2e_\lambda d_\lambda/t$. Thus we have
    $I(\sk, \A:\B|\E,\pi)\leq 2e_\lambda d_\lambda/t$.
\end{proof}

\end{document}